\documentclass[a4paper, reqno]{amsart}

\usepackage{amsfonts,amssymb,amsmath, microtype}

\usepackage[numbers,sort&compress]{natbib}
\bibpunct[, ]{[}{]}{,}{n}{,}{,}

\usepackage{graphicx}
\usepackage[shortlabels]{enumitem}
\usepackage{color}
\usepackage{mathtools}
\usepackage[caption=false]{subfig}

\usepackage{ifthen}

\theoremstyle{plain}
\newtheorem{theorem}{Theorem}[section]
\newtheorem{lemma}[theorem]{Lemma}
\newtheorem{corollary}[theorem]{Corollary}
\newtheorem{proposition}[theorem]{Proposition}

\theoremstyle{definition}
\newtheorem{definition}[theorem]{Definition}

\theoremstyle{remark}
\newtheorem{remark}{Remark}

\newboolean{nocut}
\setboolean{nocut}{true}

\usepackage{algorithm}
\usepackage{algorithmic}

\newcommand{\eset}{\emptyset}

\newcommand{\floor}[1]{\lfloor#1\rfloor}
\newcommand{\zag}[1]{\left(#1\right)}
\newcommand{\uglate}[1]{\left[#1\right]}
\newcommand{\aps}[1]{\left|#1\right|}
\newcommand{\interval}[1]{\left<#1\right>}
\newcommand{\skup}[1]{\left\{#1\right\}}
\newcommand{\norma}[1]{\left\|#1\right\|}
\newcommand{\degr}[1]{\deg\left(#1\right)}

\newcommand{\N}{\ensuremath{\mathbb{N}}}
\newcommand{\Z}{\ensuremath{\mathbb{Z}}}
\newcommand{\R}{\mathbb{R}}
\newcommand{\Rd}{\mathbb{R}^{n_1 \times n_2 \times \cdots \times n_d}}

\newcommand{\bm}[1]{\boldsymbol{#1}}
\newcommand{\mbf}[1]{\mathbf{#1}}

\DeclareMathOperator*{\conv}{conv}

\DeclareMathOperator{\rank}{rank}

\DeclareMathOperator{\vecc}{vec}
\DeclareMathOperator{\trace}{trace}

\DeclareMathOperator{\TB}{TH}
\DeclareMathOperator{\LT}{LT}

\DeclareMathOperator{\NLM}{NLM}
\DeclareMathOperator{\LM}{LM}
\DeclareMathOperator{\LCM}{LCM}
\DeclareMathOperator{\LC}{LC}
\DeclareMathOperator{\multideg}{multideg}

\DeclareMathOperator{\tr}{tr}

\begin{document}

\title{Tensor theta norms and low rank recovery}

\author{Holger Rauhut}
\address{RWTH Aachen University, Lehrstuhl C f{\"u}r Mathematik (Analysis), Pontdriesch 10, 52062 Aachen Germany}
\email{rauhut@mathc.rwth-aachen.de}

\author{\v{Z}eljka Stojanac}
\address{RWTH Aachen University, University of Cologne, Institute for Theoretical Physics, Z{\"u}lpicher Stra{\ss}e 77, 50937 Germany}
\email{stojanac@thp.uni-koeln.de}

\date{\today}

\maketitle

\begin{abstract}
We study extensions of compressive sensing and low rank matrix recovery to the recovery of tensors of low rank from incomplete linear
information. While the reconstruction of low rank matrices via nuclear norm minimization is rather well-understand by now, almost no theory is available
so far for the extension to higher order tensors due to various theoretical and computational difficulties arising for tensor decompositions.
In fact, nuclear norm minimization for matrix recovery is a tractable convex relaxation approach, but the extension of the nuclear norm to tensors
is in general NP-hard to compute. In this article, we introduce convex relaxations of the tensor nuclear norm which are computable in polynomial time via 
semidefinite programming. Our approach is based on theta bodies, a concept from computational algebraic geometry which is similar to the one 
of the better known Lasserre relaxations. We introduce polynomial ideals which are 
generated by the second order minors corresponding to different matricizations of the tensor (where the
tensor entries are treated as variables) such that the nuclear norm ball is the convex hull of the algebraic variety of the ideal. The theta body of order $k$ for
such an ideal generates a new norm which we call the $\theta_k$-norm. We show that in the matrix case, these norms reduce to the standard nuclear norm.
For tensors of order three or higher however, we indeed obtain new norms. The sequence of the corresponding unit-$\theta_k$-norm balls converges asymptotically to the unit tensor nuclear norm ball. 
By providing the Gr{\"o}bner basis for the ideals, 
we explicitly give semidefinite programs for the computation of the $\theta_k$-norm and for the minimization of the $\theta_k$-norm under
an affine constraint. Finally, numerical experiments for order-three tensor recovery via $\theta_1$-norm minimization suggest that our approach
successfully reconstructs tensors of low rank from incomplete linear (random) measurements.
\end{abstract}


\medskip


\section{Introduction and Motivation}

Compressive sensing predicts that sparse vectors can be recovered from underdetermined linear 
measurements via efficient methods such as $\ell_1$-minimization \cite{carota06,do06-2,fora13}. This finding
has various applications in signal and image processing and beyond. It has recently been observed that
the principles of this theory can be transferred to the problem of recovering a low rank matrix from underdetermined
linear measurements. One prominent choice of recovery method consists in minimizing the nuclear norm subject
to the given linear constraint \cite{fa02-2,fapare10}. 
This convex optimization problem can be solved efficiently and recovery results for certain
random measurement maps have been provided, which quantify the minimal number of measurements 
required for successful recovery \cite{fapare10,care09,capl11,beeiflgrli09,gr11,kurate14}. 

There is significant interest in going one step further and to extend the theory to the recovery of low rank tensors 
(higher-dimensional arrays) from incomplete linear measurements. Applications include image and video inpainting \cite{limuwoye09},
reflectance data recovery \cite{limuwoye09} (e.g.\ for use in photo-realistic raytracers),
machine learning \cite{aubiporo13}, and seismic data processing \cite{krsa12}.
Several approaches have 
already been introduced \cite{limuwoye09,gareya11,gohumuwr14,rascst13,rascst15}, but 
unfortunately, so far, for none of them a completely satisfactory theory is available. 
Either the method is not tractable \cite{yuzh14}, or
no (complete) rigorous recovery results quantifying the minimal number of measurements are available \cite{gareya11,limuwoye09,rascst13,rascst15,Karlsson.Kressner.Uschmajew:2014,da2013hierarchical,kressner2014low}, or the available
bounds are highly nonoptimal \cite{gohumuwr14,DuarteBaraniuk11,liu2014generalized}. 
For instance, the computation (and therefore, also the minimization) of the tensor nuclear norm (\cite{defant1992tensor,ryan2002introduction,wong1979schwartz}) for higher order tensors
is in general NP-hard \cite{frli14} -- 
nevertheless, some recovery results for tensor completion via nuclear norm minimization are available in \cite{yuzh14}. 
Moreover, versions of iterative hard thresholding for various tensor formats have been introduced \cite{rascst13,rascst15}.
This approach leads to a computationally tractable algorithm, which empirically works well. 
However, only a partial analysis based on the tensor restricted isometry property has been provided, which so far only
shows convergence under a condition on the iterates 
that cannot be checked a priori. Nevertheless, the tensor restricted isometry property (TRIP)
has been analyzed for certain random measurement maps \cite{rascst13,rascst15,DBLP:journals/corr/RauhutSS16}. These near optimal bounds on the number of measurements
ensuring the TRIP, however, provide only a hint on how many measurements are required because the link between
the TRIP and recovery is so far only partial \cite{rascst15,DBLP:journals/corr/RauhutSS16}.

This article introduces 
a new approach for tensor recovery based on convex relaxation. The idea is to further relax the nuclear norm
in order to arrive at a norm which can be computed (and minimized under a linear constraint) in polynomial time.
The hope is that the new norm is only a slight relaxation and possesses very similar properties as the nuclear norm.
Our approach is based on  theta bodies, a concept from computational algebraic
geometry \cite{lo79-1,gopath10,blpath13} which is similar to the better known Lasserre relaxations \cite{la10}. 
We arrive at a whole family of convex bodies (indexed by a polynomial degree), which form
convex relaxations of the unit nuclear norm ball. 
The resulting norms are called  theta norms. The corresponding unit norm balls are nested and contain the unit nuclear norm ball.
Even more, the sequence of the unit-$\theta_k$-norm balls converges asymptotically to the unit tensor nuclear norm ball.
They can be computed by semidefinite optimization, and also the minimization of the $\theta_k$-norm subject to a linear constraints is a semidefinite program (SDP) whose solution can be computed in
polynomial time -- the complexity growing with $k$.

The basic idea for the construction of these new norms is to define polynomial ideals, where each variable corresponds to an entry of the tensor,
such that its algebraic variety consists of the rank-one tensors of unit Frobenius norm. The convex hull of this set is the tensor nuclear norm ball.
The ideals that we propose are generated by the minors of order two of all matricizations of the tensor (or at least of a subset of the possible matricizations)
together with the polynomial corresponding to the squared Frobenius norm minus one. Here, a matricization denotes a matrix which
is generated from the tensor by combining several indices to a row index, and the remaining indices to a column index. 
In fact, all such minors being zero simultaneously 
means that the tensor has rank one. The $k$-theta body of the ideal 
corresponds then to a relaxation of the convex hull of its algebraic variety, i.e., to a further relaxation of the tensor nuclear norm. The index $k \in \N$
corresponds to a polynomial degree involved in the construction of the theta bodies (some polynomial is required to be $k$-sos modulo the ideal, see below), 
and $k = 1$ leads to the largest theta body in a family
of convex relaxations.

We will show that for the matrix case (tensors of order $2$), our approach does not lead to new norms. All resulting theta norms are rather equal to the
matrix nuclear norm. This fact suggests that the theta norms in the higher order tensor case
are all natural generalizations of the matrix nuclear norm. 

We derive the corresponding semidefinite programs explicitly and present numerical experiments which show that
$\theta_1$-norm minimization successfully recovers tensors of low rank from few random linear measurements.
Unfortunately, a rigorous theoretical analysis of the recovery performance of $\theta_k$-minimization is not yet available
but will be the subject of future studies.

\subsection{Low rank matrix recovery}

Before passing to tensor recovery, we recall some basics on matrix recovery. 
Let $\mbf{X} \in \mathbb{R}^{n_1 \times n_2}$ of rank at most $r \ll \min\{n_1,n_2\}$, and suppose we are given linear measurements
\[
\mbf{y} = \mathcal{A}(\mathbf{X}),
\]
where $\mathcal{A} : \R^{n_1 \times n_2} \to \R^m$ is a linear map with $m \ll n_1  n_2$.
Reconstructing $\mathbf{X}$ from $\mbf{y}$ amounts to solving an underdetermined linear system. Unfortunately,
the rank minimization problem of computing the minimizer of
\[
\min_{\mbf{Z} \in \R^{n_1 \times n_2}} \operatorname{rank}(\mathbf{Z})   \quad \mbox{ subject to } \mathcal{A}(\mathbf{Z}) = \mbf{y}
\]
is NP-hard in general. As a tractable alternative, the convex optimization problem
\begin{equation}\label{nucl:norm}
\min_{\mbf{Z} \in \R^{n_1 \times n_2}} \| \mathbf{Z} \|_*
\quad \mbox{ subject to } \mathcal{A}(\mathbf{Z}) = \mbf{y}
\end{equation}
has been suggested \cite{fa02-2,fapare10}, where the nuclear norm $\|\mbf{Z} \|_* = \sum_j \sigma_j(\mbf{Z})$  is the sum of the singular values
of $\mbf{Z}$. This problem can be solved efficiently by various methods \cite{bova04}. For instance, it can be reformulated 
as a semidefinite program \cite{fa02-2}, but splitting methods may be more efficient \cite{pabo13,toyu10}. 

While it is hard to analyze low rank matrix recovery for deterministic measurement maps, optimal bounds are available
for several random matrix constructions. If $\mathcal{A}$ is a Gaussian measurement map, i.e., 
\[
\mathcal{A}(\mbf{X})_j = \sum_{k,\ell} \mathcal{A}_{j k\ell} X_{k\ell}, \quad j \in [m] := \{1,2,\hdots,m\},
\]
where the $\mathcal{A}_{jkl}$, $j \in [m], k \in [n_1], \ell \in [n_2]$, are independent mean-zero, variance one Gaussian random
variables, then a matrix $\mbf{X}$ of rank at most $r$ can be reconstructed exactly 
from $\mbf{y} = \mathcal{A}(\mbf{X})$  via nuclear norm minimization \eqref{nucl:norm} with probability at least $1-e^{-c m}$ 
provided that 
\begin{equation}\label{bound:mtx:case}
m \geq C r n, \qquad n = \max\{n_1,n_2\},
\end{equation}
where the constants $c, C> 0$ are universal \cite{capl11,chparewi10}. 
Moreover, the reconstruction is stable under passing to only approximately low rank matrices and under adding 
noise on the measurements. Another interesting measurement map corresponds to the matrix completion problem \cite{care09,cata10-2,gr11,JMLR:v16:chen15b}, 
where the measurements are randomly chosen entries of the matrix $\mathbf{X}$.
Measurements taken as Frobenius inner products with rank-one matrices are studied in \cite{kurate14}, and arise in the phase
retrieval problem as special case \cite{castvo13}. Also here, $m \geq C rn$ (or $m \geq Cr n \log(n)$ for certain structured measurements) 
is sufficient for exact recovery. 

\subsection{Tensor recovery}

An order-$d$ tensor (or mode-$d$-tensor) is an element $\mbf{X} \in \R^{n_1 \times n_2 \times \cdots \times n_d}$ 
indexed by $[n_1]\times [n_2] \times \cdots \times [n_d]$. Of course, the case $d=2$ corresponds to matrices.
For $d \geq 3$, several notions and computational tasks become much more involved than for the matrix case.
	Already the notion of rank requires some clarification, and in fact, several different definitions are available, see for instance \cite{SAPM:SAPM192761164,hitchcock-rank-1927,landsbergtensors,grasedyck2011introduction}. 
We will mainly work with the  canonical rank or CP-rank in the following.
A $d$th-order tensor $\mbf{X} \in \R^{n_1 \times n_2 \times \cdots \times n_d}$ 
is of rank one if there exist vectors $\mbf{u}^1 \in \R^{n_1},  \mbf{u}^2 \in \R^{n_2}, \ldots, \mbf{u}^d \in \R^{n_d}$ such that 
$\mbf{X}=\mbf{u}^1 \otimes \mbf{u}^2 \otimes \cdots \otimes \mbf{u}^d$ or elementwise
$$
X_{i_1 i_2 \ldots i_d}=u_{i_1}^1 u_{i_2}^2 \cdots u_{i_d}^d. 
$$
The CP-rank (or canonical rank and in the following just rank) of a 
tensor $\mbf{X} \in \R^{n_1 \times n_2 \times \cdots \times n_d}$, similarly as in the matrix case, 
is the smallest number of rank-one tensors that sum up to $\mbf{X}$.

Given a linear measurement map $\mathcal{A} : \R^{n_1 \times \cdots \times n_d} \to \R^m$ (which can represented as
a $(d+1)$th-order tensor), our aim is to recover a tensor $\mbf{X} \in \R^{n_1 \times \cdots \times n_d}$ from
$\mbf{y} = \mathcal{A}(\mbf{X})$ when $m \ll n_1 \cdot n_2 \cdots n_d$. The matrix case $d=2$ suggests to consider
minimization of the tensor nuclear norm for this task,
\[
\min_{\mbf{Z}} \|\mbf{Z}\|_* \quad \mbox{ subject to } {\mathcal{A}}(\mbf{Z}) = \mbf{y},
\]
where the nuclear norm is defined as
\begin{align*}
\norma{\mbf{X}}_*=\min \Big\{\sum_{k=1}^r \aps{c_k}:  \mbf{X} =&\sum_{k=1}^r c_k \mbf{u}^{1,k} \otimes \mbf{u}^{2,k} \otimes \cdots \otimes \mbf{u}^{d,k}, r \in \mathbb{N},  \\
&\hspace{2cm}\norma{\mbf{u}^{i,k}}_{\ell_2}=1,  
i \in \uglate{d}, k \in \uglate{r}\Big\}.
\end{align*}
Unfortunately, in the tensor case, computing the canonical rank of a tensor, as well as computing the nuclear norm of a tensor is NP-hard in general, see \cite{ha90-2,hili13,frli14}. Let us nevertheless mention that some theoretical results for tensor
recovery via nuclear norm minimization are contained in \cite{yuzh14}.

We remark that, unlike in the matrix scenario, the tensor rank and consequently  the tensor nuclear norm are dependent on the choice of base field, see for example \cite{Brylinski,de2008tensor,frli14}. In other words, the rank (and the nuclear norm) of a given  tensor with real entries depends on whether we regard it as a real tensor or as a complex tensor. In this paper, we focus only on tensors with real-valued entries, i.e., we work over the field $\R$.

The aim of this article is to introduce relaxations of the tensor nuclear norm, based on  theta bodies, which
is both computationally tractable and whose minimization allows for exact recovery of low rank tensors from incomplete
linear measurements. 

Let us remark that one may reorganize (flatten) a low rank tensor $\mbf{X} \in \R^{n \times n \times n}$ 
into a low rank matrix $\tilde{\mbf{X}} \in \R^{n \times n^2}$ 
and simply apply concepts from matrix recovery. However, the bound \eqref{bound:mtx:case} on the required number of
measurements then reads
\begin{equation}\label{m:bound:weak}
m \geq C r n^2.
\end{equation}
Moreover, it has been suggested in \cite{gareya11,hakato10,limuwoye09} to minimize the sum of nuclear norms of 
the unfoldings (different reorganizations of the tensor 
as a matrix) subject to the linear constraint matching the measurements. Although this seems to be a reasonable approach at first sight,
it has been shown in \cite{elfajaoy12}, that it cannot work with less measurements than stated by the estimate in \eqref{m:bound:weak}. This is essentially due to the fact that the tensor structure is not represented. That is, instead of solving a tensor nuclear norm minimization problem under the assumption that the tensor is of low rank, the matrix nuclear norm minimization problem is being solved under the assumption that a particular matricization of a tensor is of low rank.

Bounds for a version of the restricted isometry property for certain tensor formats in \cite{DBLP:journals/corr/RauhutSS16} suggest that
\[
m \geq C r^2 n
\]
measurements should be sufficient when working directly with the tensor structure -- precisely, this bound
uses the  tensor train format \cite{os11}. (Possibly, the term $r^2$ may even be lowered to $r$ when using the ``right'' tensor format.) 
However, connecting the
restricted isometry property in a completely 
satisfactory way with the success of an efficient tensor recovery algorithm is still open. (Partial results are contained
in \cite{DBLP:journals/corr/RauhutSS16}.)
In any case, this suggests that one should exploit the tensor structure of the problem rather than reducing to a matrix
recovery problem in order to recover a low rank tensor using the minimal number of measurements.
Of course, similar considerations apply to tensors of order higher than three, where the difference between
the reduction to the matrix case and working directly with the tensor structure will become even stronger.

Unlike in the previously mentioned contributions, we consider the canonical tensor rank and the corresponding tensor nuclear norm, which respects the tensor structure. Even more, it is expected that the bound on the minimal number of measurements needed for low rank tensor recovery via tensor nuclear norm minimization is optimal, see also \cite{yuzh14}, where tensor completion via tensor nuclear norm minimization has been considered. Unfortunately, it is in general NP-hard to solve this optimization problem (since it is NP-hard to compute the tensor nuclear norm). To overcome this difficulty, in this paper, we provide the tensor $\theta_k$-norms -- the new tensor norms which can be computed via semidefinite programming. These norms are tightly related to the tensor nuclear norm. That is, the unit $\theta_k$-norm balls (which are defined for $k\in \mathbb{N}$) satisfy 
\begin{align*}
\skup{\mbf{X}: \norma{\mbf{X}}_{\theta_1}\leq 1} \supseteq \cdots  \supseteq \skup{\mbf{X}: \norma{\mbf{X}}_{\theta_{k}}\leq 1}& \supseteq \skup{\mbf{X}: \norma{\mbf{X}}_{\theta_{k+1}}\leq 1} \\& \supseteq \cdots \supseteq \skup{\mbf{X}: \norma{\mbf{X}}_{*}\leq 1}.
\end{align*}
In particular, we show that in the matrix scenario all $\theta_k$-norms coincide with the matrix nuclear norm. In case of order-$d$ tensors ($d \geq 3$), we prove that the sequence of the unit-$\theta_k$-norm balls converges asymptotically to the unit tensor nuclear norm ball. Next, we provide numerical experiments on low rank tensor recovery via $\theta_1$-norm minimization. 
We provide numerical experiments for $\theta_1$-minimization that indicate that this is a very promising approach for low rank tensor recovery. However, we note that standard solvers for semidefinite programs only allow us to test our method on small to moderate size problems. Nevertheless, it is likely that specialized efficient algorithms can be developed. Indeed, recall that $\theta_k$-norms all coincide with the matrix nuclear norm and the state-of-the-art algorithms allow us computing the nuclear norm of matrices of large dimensions. This suggests the possibility that new algorithms could be developed which would allow us to apply our method on larger tensors. Thus, this paper presents the first step in a new convex optimization approach to low rank tensor recovery.

\subsection{Some notation}

We write vectors with small bold letters, matrices and tensors with capital bold letters and sets with capital calligraphic letters. The cardinality of a set $\mathcal{S}$ is denoted by $\aps{\mathcal{S}}$.

For a matrix $\mbf{A} \in \R^{m \times n}$ and subsets $\mathcal{I} \subset \uglate{m}$, $\mathcal{J} \subset \uglate{n}$ 
the submatrix of $\mathbf{A}$ with columns indexed by $\mathcal{I}$ and rows indexed by $\mathcal{J}$ is denoted
by $\mbf{A}_{\mathcal{I}, \mathcal{J}}$.
A set of all order-$k$ minors of $\mbf{A}$ is of the form
$$ \skup{\det(\mbf{A}_{\mathcal{I},\mathcal{J}}): \mathcal{I} \subset \uglate{m}, \mathcal{J} \subset \uglate{n}, \aps{\mathcal{I}}=\aps{\mathcal{J}}=k }.$$
The Frobenius norm of a matrix $\mbf{X} \in \R^{m \times n}$ is given as
$$\norma{\mbf{X}}_F= \sqrt{\sum_{i=1}^m \sum_{j=1}^n X_{ij}^2}=\sqrt{\sum_{i=1}^{\min \{m,n\}} \sigma_i^2},$$
where the $\sigma_i$ list the singular values of $\mbf{X}$. The nuclear norm is given by 
$\norma{\mbf{X}}_*=\sum_{i=1}^{\min \{m,n\}} \sigma_i$. It is well-known that its unit ball is the convex
hull of all rank-one matrices of unit Frobenius norm.

The vectorization of a tensor $\mbf{X} \in \R^{n_1 \times n_2 \times \cdots \times n_d}$ is denoted by
$\vecc(\mbf{X}) \in\R^{n_1 n_2 \cdots n_d}$. The ordering of the elements in $\vecc(\mbf{X})$ 
is not important as long as it remains consistent. Fibers are a higher order analogue of matrix rows and columns. For $k \in \uglate{d}$, the mode-$k$ fiber of a $d$th-order tensor is obtained by fixing every index except for the $k$-th one. 
The Frobenius norm of a $d$th-order tensor $\mbf{X} \in \R^{n_1 \times n_2 \times \cdots \times n_d}$ is defined as
$$ \norma{\mbf{X}}_F=\sqrt{\sum_{i_1=1}^{n_1} \sum_{i_2=1}^{n_2}  \cdots \sum_{i_d=1}^{n_d} X_{i_1 i_2 \cdots i_d}^2}.$$
Matricization (also called flattening) is the operation that transforms a tensor into a matrix. More precisely, 
for a $d$th-order tensor $\mbf{X} \in \R^{n_1 \times n_2 \times \cdots \times n_d}$ and an ordered subset 
$\mathcal{S} \subseteq \uglate{d}$, an $\mathcal{S}$-matricization 
$\mbf{X}^{\mathcal{S}} \in \R^{\prod_{k \in \mathcal{S}} n_k \times \prod_{\ell \in \mathcal{S}^c}n_{\ell}}$ is defined as
$${X}^{\mathcal{S}}_{(i_k)_{k \in \mathcal{S}}, (i_\ell)_{\ell \in \mathcal{S}^c}}={X}_{i_1 i_2 \ldots i_d}, $$
i.e., the indexes in the set $\mathcal{S}$ define the rows of a matrix and the indexes in the set 
$\mathcal{S}^c=\uglate{d}\backslash \mathcal{S}$ define the columns. 
For a singelton set $\mathcal{S}=\{i\}$, for $i \in \uglate{d}$, we call the $\mathcal{S}$-matricization the $i$-th unfolding. 
Notice that every $\mathcal{S}$-matricization of a rank-one tensor is a rank-one matrix. Conversely, if every 
$\mathcal{S}$-matricization of a tensor is a rank-one matrix, then the tensor is of rank one. This is even true, if all unfoldings of a tensor
are of rank one.

We often use MATLAB notation. Specifically, for a $d$th-order tensor 
$\mbf{X} \in \R^{n_1 \times n_2 \times \cdots \times n_d}$, we write $\mbf{X}(:,:,\ldots,:,k) $ for the $(d-1)$-order subtensor 
in $\R^{n_1 \times \cdots \times n_{d-1}}$ obtained by fixing the last index $\alpha_d$ to $k$.
For simplicity, the subscripts $\alpha_1\alpha_2\cdots \alpha_d$ and $\beta_1\beta_2 \cdots \beta_d $ will often be denoted by $\bm{\alpha}$ and $\bm{\beta}$, respectively. In particular, instead of writing $x_{\alpha_1\alpha_2\ldots \alpha_d}x_{\beta_1\beta_2\ldots \beta_d}$, we often just write $x_{\bm{\alpha}}x_{\bm{\beta}}$. Below, we will use the grevlex
ordering of monomials indexed by subscripts $\bm{\alpha}$, which in particular requires to define an ordering for such subscripts.
We make the agreement that $x_{11\ldots 11} > x_{11 \ldots 12} > \cdots > x_{11\ldots 1 n_d} > x_{111 \ldots 21}> \ldots > x_{n_1 n_2 \ldots n_d}$. 

\subsection{Structure of the paper}

In Section~\ref{IntrodTheta} we will review the basic definition and properties of theta bodies.
Section~\ref{sec:matrix-case} considers the matrix case. We introduce a suitable polynomial ideal whose algebraic variety is the 
set of rank-one unit Frobenius norm matrices. We discuss the corresponding $\theta_k$-norms and show that they all coincide with
the matrix nuclear norm. The case of $2 \times 2$-matrices is described in detail. In Section~\ref{sec:tensorTH} we pass to the tensor case and discuss first
the case of order-three tensors. We introduce a suitable polynomial ideal, provide its reduced Gr{\"o}bner basis and define the corresponding
$\theta_k$-norms. We additionally show that considering matricizations corresponding to the TT-format will lead to the same polynomial ideal and thus to the same $\theta_k$-norms. 
The general $d$th-order case is discussed at the end of Section~\ref{sec:tensorTH}. Here, we define the 
polynomial ideal $J_d$ which corresponds to the set of all possible matricizations of the tensor. We show that a certain set of order-two minors forms 
the reduced Gr{\"o}bner basis for this ideal, which is key for defining the $\theta_k$-norms. We additionally show that polynomial ideals corresponding to different tensor formats (such as TT format or Tucker/HOSVD format) coincide with the ideal $J_d$ and consequently, they  lead to the same $\theta_k$-norms.
In Section~\ref{Section:Convergence} we discuss the convergence of the sequence of the unit-$\theta_k$-norm balls to the unit tensor nuclear norm ball. 
Section~\ref{sec:complexity} briefly discusses the polynomial runtime of the algorithms for computing and minimizing the $\theta_k$-norms showing
that our approach is tractable. Numerical experiments for low rank recovery of third-order tensors are presented in Section~\ref{sec:numerics}, which
show that our approach successfully recovers a low rank tensor from incomplete Gaussian random measurements.
Appendix~\ref{sec:appA} discusses some background from computer algebra (monomial orderings and Gr{\"o}bner bases) that
is required throughout the main body of the article.

\section{Theta bodies} \label{IntrodTheta}

As outlined above, we will introduce new tensor norms as relaxations of the nuclear norm in order to come up
with a new convex optimization approach for low rank tensor recovery. Our approach builds on theta bodies, a recent concept from computational algebraic geometry, which is similar to Lasserre relaxations \cite{la10}.
In order to introduce it, we first discuss the necessary basics from computational commutative algebra. For more information, 
we refer to \cite{colios05,colios07} and to the appendix.

For  a non-zero polynomial $f=\sum_{\bm{\alpha}}a_{\bm{\alpha}} {\mbf{x}}^{\bm{\alpha}}$  in $\R\uglate{\mbf{x}}=\R\uglate{x_1,x_2,\ldots,x_n}$ and  a monomial order~$>$,  
we denote
\begin{enumerate}[noitemsep]
\item[a)] the multidegree of $f$ by
$ \multideg\zag{f}=\max\zag{\bm{\alpha} \in \Z_{\geq 0}^n: a_{\bm{\alpha}} \neq 0} $,
\item[b)] the leading coefficient of $f$ by
$ \LC\zag{f}=a_{\multideg\zag{f}} \in \R$,
\item[c)] the leading monomial of $f$ by 
$ \LM\zag{f}=\mbf{x}^{\multideg\zag{f}}$,
\item[d)] the leading term of $f$ by
$ \LT\zag{f}=\LC\zag{f} \LM\zag{f}.$
\end{enumerate}
Let $J \subset  \R\uglate{\mbf{x}}$ be a polynomial ideal. 
Its real algebraic variety 
is the set of all points in $\mbf{x} \in \R^n$ where all polynomials in the ideal vanish, i.e.,
$$\nu_{\R}\zag{J}=\{\mbf{x} \in \R^n: f(\mbf{x})=0, \, \text{ for all } f \in J \}.$$
By Hilbert's basis theorem \cite{colios07} every polynomial ideal in $\R\uglate{\mbf{x}}$ has a finite generating set.
Thus, we may assume that $J$ is generated by a set $\mathcal{F}=\skup{f_1,f_2,\ldots,f_k}$ of polynomials in $\R\uglate{\mbf{x}}$ and write
$$J=\interval{f_1,f_2,\ldots,f_k}=\interval{\skup{f_i}_{i \in \uglate{k}}} \quad \text{or simply} \quad J=\interval{\mathcal{F}}.$$
Its real algebraic variety is the set
$$\nu_{\R}\zag{J}=\{\mbf{x} \in \R^n: f_i(\mbf{x})=0 \, \text{ for all } i \in \uglate{k} \}.$$
Throughout the paper, $\R\uglate{\mbf{x}}_k$ denotes the set of polynomials of degree at most $k$. A degree one polynomial
is also called linear polynomial.
A very useful certificate for positivity of polynomials is contained in the following definition \cite{gopath10}.

\begin{definition}
Let $J$ be an ideal in $\R\uglate{\mbf{x}}$. A polynomial $f \in \R\uglate{\mbf{x}}$ is \textit{$k$-sos mod $J$} if there exists 
a finite set of polynomials $h_1,h_2,\ldots,h_t \in \R\uglate{\mbf{x}}_k$ such that $f \equiv \sum_{j=1}^t h_j^2$ mod $J$, i.e., if $f-\sum_{j=1}^t h_j^2 \in J$.
\end{definition}

A special case of theta bodies was first introduced by Lov\'{a}sz in \cite{lo79-1} and in full generality they appeared in \cite{gopath10}. 
Later, they have been analyzed in \cite{gouveia2009new,grande2014theta}. The definitions and theorems in the remainder 
of the section are taken from \cite{gopath10}. 

\begin{definition}[Theta body] 
Let $J \subseteq \R\uglate{\mbf{x}}$ be an ideal. For a positive integer $k$, the \textit{$k$-th theta body of $J$} is defined as
\begin{equation*}
\TB_k\zag{J}:=\skup{\mbf{x} \in \R^n : f\zag{\mbf{x}}\geq 0 \text{ for every linear } f \text{ that is $k$-sos mod }J }.
\end{equation*}
We say that an ideal $J \subseteq \R\uglate{\mbf{x}}$ is \textit{$\TB_k$-exact} if $\TB_k\zag{J}$ equals  $\overline{\conv\zag{\nu_{\R}\zag{J}}}$, the closure of the convex hull of $\nu_{\R}\zag{J}$. 
\end{definition}
Theta bodies are closed convex sets, while $\conv\zag{\nu_{\R}\zag{J}}$ may not necessarily be closed and by definition, 
\begin{equation} \label{seqtheta}
\TB_1\zag{J} \supseteq \TB_2\zag{J} \supseteq \cdots \supseteq \conv\zag{\nu_{\R}\zag{J}}.
\end{equation}
The theta-body sequence of $J$ can converge (finitely or asymptotically), if at all, only to 
$\overline{\conv\zag{\nu_{\R}\zag{J}}}$. More on guarantees on convergence can be found in \cite{gopath10,grande2014theta}. However, to our knowledge, none of the existing guarantees apply to the cases discussed below. 

Given any polynomial, it is possible to check whether it is $k$-sos mod $J$ using a Gr{\"o}bner basis and semidefinite programming. However, using this definition in practice requires knowledge of all linear polynomials (possibly infinitely many) that are $k$-sos mod $J$. To overcome this difficulty, we need an alternative description of $\TB_k\zag{J}$ discussed next.

As in \cite{blpath13}, we assume that there are no linear polynomials in the ideal $J$. Otherwise, some variable $x_i$ would be congruent to a linear combination of other variables modulo $J$ and we could work in a smaller polynomial ring $\R\uglate{\mbf{x}^i}=\R\uglate{x_1,x_2,\ldots,x_{i-1},x_{i+1},\ldots,x_n}$. Therefore, $\R\uglate{\mbf{x}}_1/J \cong \R\uglate{\mbf{x}}_1$ and $\{1+J, x_1+J, \ldots, x_n+J\}$ can be completed to a basis $\mathcal{B}$ of $\R\uglate{\mbf{x}}/J$. Recall that the degree of an equivalence class $f+J$, denoted by $\degr{f+J}$, is the smallest degree of an element in the class. We assume that each element in the basis $\mathcal{B}=\{f_i+J\}$ of $\R\uglate{\mbf{x}}/J$ is represented by the polynomial whose degree equals the degree of its equivalence class, i.e., $\degr{f_i+J}=\degr{f_i}$. In addition, we assume that $\mathcal{B}=\skup{f_i+J}$ is ordered so that $f_{i+1} > f_{i}$, where $>$ is a fixed monomial ordering. Further, we define the set $\mathcal{B}_k$ 
\begin{equation*}
\mathcal{B}_k:=\{ f+J \in \mathcal{B}: \deg(f+J) \leq k\}.
\end{equation*}

\begin{definition}[Theta basis]\label{thetabasis}
Let $J \subseteq \R\uglate{\mbf{x}}$ be an ideal. A basis $\mathcal{B}=\{f_0+J, f_1+J, \ldots\}$ of $\R\uglate{\mbf{x}}/J$ is a \textit{$\theta$}-basis if it has the following properties
\begin{enumerate}
\item[1)] $\mathcal{B}_1=\skup{1+J, x_1 +J, \ldots, x_n+J}$,
\item[2)] if $\degr{f_i+J}, \degr{f_j+J} \leq k$ then $f_i f_j + J$ is in the $\R$-span of $\mathcal{B}_{2k}$.
\end{enumerate}
\end{definition}
As in \cite{blpath13,gopath10} we consider only monomial bases $\mathcal{B}$ of $\R\uglate{\mbf{x}}/J$, i.e., bases 
$\mathcal{B}$ such that $f_i$ is a monomial, for all $f_i+J \in \mathcal{B}$.

For determining a $\theta$-basis, we first need to compute the reduced Gr{\"o}bner basis $\mathcal{G}$ of the ideal $J$, see Definitions~\ref{groebner} and \ref{reducedGr}. The set ${\mathcal{B}}$ will satisfy the second property in the definition of the theta basis if the reduced Gr{\"o}bner basis is with respect to an ordering which first compares the total degree. Therefore, throughout the paper we use the graded reverse monomial ordering (Definition~\ref{def:grevlex}) or simply grevlex ordering, although also the graded lexicographic ordering would be appropriate.

A technique to compute a $\theta$-basis $\mathcal{B}$ of $\R\uglate{\mbf{x}}/J$ consists in taking $\mathcal{B}$ to be 
the set of equivalence classes of the standard monomials of the corresponding initial ideal 
$$J_{\text{initial}}=\interval{\skup{\LT(f)}_{f \in J}}=\interval{\skup{\LT(g_i)}_{i \in \uglate{s}}}, $$
where $\mathcal{G}=\interval{g_1,g_2,\ldots,g_s}$ is the reduced Gr{\"o}bner basis of the ideal $J$. In other words, a set $\mathcal{B}=\skup{f_0+J,f_1+J,\ldots}$ will be a $\theta$-basis of $\R\uglate{\mbf{x}}/J$ if it contains all $f_i+J$ such that \begin{enumerate}
\item[1)] $f_i$ is a monomial
\item[2)] $f_i$ is not divisible by any of the monomials in the set $\skup{\LT(g_i): i \in \uglate{s}}$.
\end{enumerate}

The next important tool we need is the  combinatorial moment matrix of $J$. To this end, we
fix a $\theta$-basis $\mathcal{B}=\skup{f_i +J}$ of $\R\uglate{\mbf{x}}/J$ and define $\uglate{\mbf{x}}_{\mathcal{B}_k}$ to be the column vector formed by all elements of $\mathcal{B}_k$ in order. Then $\uglate{\mbf{x}}_{\mathcal{B}_k}\uglate{\mbf{x}}_{\mathcal{B}_k}^T$ is a square matrix indexed by $\mathcal{B}_k$ and its $\zag{i,j}$-entry is equal to $f_i f_j + J$. By hypothesis, the entries of $\uglate{\mbf{x}}_{\mathcal{B}_k}\uglate{\mbf{x}}_{\mathcal{B}_k}^T$ lie in the $\R$-span of $\mathcal{B}_{2k}$. Let $\{\lambda_{i,j}^l\}$ be the unique set of real numbers such that $f_i f_j + J=\sum_{f_l + J \in \mathcal{B}_{2k}}\lambda_{i,j}^l \zag{f_l+J}$.

The theta bodies can be characterized via the combinatorial moment matrix as stated in the next result from
\cite{gopath10}, which will be the basis for computing and minimization the new tensor norm introduced below via
semidefinite programming.

\begin{definition}\label{def:moment}
Let $J, \mathcal{B}$ and $\{\lambda_{i,j}^l\}$ be as above. Let $\mbf{y}$ be a real vector indexed by $\mathcal{B}_{2k}$ with $y_0=1$, where $y_0$ is the first entry of $\mbf{y}$, indexed by the basis element $1+J$. The \textit{$k$-th combinatorial moment matrix} $\mbf{M}_{\mathcal{B}_k}\zag{\mbf{y}}$ of $J$ is the real matrix indexed by $\mathcal{B}_k$ whose $\zag{i,j}$-entry is $\uglate{\mbf{M}_{\mathcal{B}_k}\zag{\mbf{y}}}_{i,j}=\sum_{f_l+J \in \mathcal{B}_{2k}} \lambda_{i,j}^l y_l$. 
\end{definition}

\begin{theorem}\label{closureQ}
The $k$-th theta body of $J$, $\TB_k\zag{J}$, is the closure of
\begin{equation*}
\mbf{Q}_{\mathcal{B}_k}\zag{J}=\pi_{\R^n} \skup{\mbf{y} \in \R^{\mathcal{B}_{2k}} : \mbf{M}_{\mathcal{B}_k}\zag{\mbf{y}} \succeq 0, y_0=1},
\end{equation*}
where $\pi_{\R^n}$ denotes the projection onto the variables $y_1=y_{x_1 + J},\ldots,y_n=y_{x_n + J}$.
\end{theorem}

\ifnocut
Algorithm \ref{CompTB1} shows a step-by-step procedure for computing $\TB_k(J)$.

\begin{algorithm}
\caption{\normalsize Algorithm for computing $\TB_k(J)$}
\label{CompTB1}
\begin{algorithmic}
\REQUIRE  An ideal $J \in \R\uglate{\mbf{x}}=\R\uglate{x_1,x_2,\ldots,x_n}$.
\STATE \textbf{Compute} the reduced Gr{\"o}bner basis for the ideal $J$  
\STATE \textbf{Compute} a $\theta$-basis $\mathcal{B}=\mathcal{B}_1 \cup \mathcal{B}_2 \cup \ldots=\{f_0+J,f_1+J,\ldots\}$ of $\R\uglate{\mbf{x}}/J$ (see Definition \ref{thetabasis})
\STATE \textbf{Compute} the combinatorial moment matrix $\mbf{M}_{\mathcal{B}_k}\zag{\mbf{y}}$: 
\begin{enumerate}[noitemsep,nolistsep]
\item $\uglate{\mbf{x}}_{\mathcal{B}_k}=\skup{\text{all elements of }\mathcal{B}_k \text{ in order}}$
\item $\zag{\mbf{X}_{\mathcal{B}_k}}_{i,j}=\zag{\uglate{\mbf{x}}_{\mathcal{B}_k} \uglate{\mbf{x}}_{\mathcal{B}_k}^T}_{i,j}=f_if_j+J=\sum_{f_l+J \in \mathcal{B}_{2k}}\lambda_{i,j}^l \zag{f_l+J}$
\item $\uglate{\mbf{M}_{\mathcal{B}_{k}}\zag{\mbf{y}}}_{i,j}=\sum_{f_l+J \in \mathcal{B}_{2k}} \lambda_{i,j}^l y_l$ 
\end{enumerate} 
\ENSURE $\TB_k\zag{J}$ is the closure of 
\begin{center}
$  \mbf{Q}_{\mathcal{B}_k}\zag{J}=\pi_{\R^n}\skup{\mbf{y} \in \R^{\mathcal{B}_{2k}}: \mbf{M}_{\mathcal{B}_k}\zag{\mbf{y}}\succeq 0, y_0=1}.$ 
\end{center}
\medskip
\end{algorithmic}
\end{algorithm}
\fi
\section{The matrix case}
\label{sec:matrix-case}

As a start, we consider the matrix nuclear unit norm ball 
and provide hierarchical relaxations via theta bodies. 
The $k$-th relaxation defines a matrix unit $\theta_k$-norm ball with the property 
$$
\norma{\mbf{X}}_{\theta_k} \leq \norma{\mbf{X}}_{\theta_{k+1}} \quad \text{ for all } \mbf{X} \in \R^{m \times n}  \text{ and all } k \in \N .
$$ 
However, we will show that all these $\theta_k$-norms coincide with the matrix  nuclear norm.

The first step in computing hierarchical relaxations of the unit nuclear norm ball consists in finding 
a polynomial ideal $J$ such that its algebraic variety (the set of points for which the  ideal vanishes) 
coincides with the set of all rank-one, unit Frobenius norm matrices
\begin{equation}\label{variety:rankone} 
\nu_{\R}(J) = \skup{\mbf{X}  \in \R^{m \times n}: \norma{\mbf{X}}_F=1, \, \rank\zag{\mbf{X}}=1}.
\end{equation}
Recall that the convex hull of this set is the nuclear norm ball.
The following lemma states the elementary fact that a non-zero matrix is a rank-one matrix if and only if all its minors of order two are zero. 

For notational purposes, we define the following polynomials in $\R\uglate{\mbf{x}}=\R[x_{11},x_{12},$ $\ldots,x_{mn}]$
\begin{align}
 g(\mbf{x})=\sum_{i=1}^m \sum_{j=1}^n x_{ij}^2-1 \text{ and }& f_{ijkl}(\mbf{x})=x_{il}x_{kj}-x_{ij}x_{kl} \nonumber\\ & \quad \text{ for } 1\leq i<k \leq m, \, 1\leq j <l \leq n. \label{def:f,g}
\end{align}

\begin{lemma}\label{rankone}
  Let $\mbf{X} \in \R^{m \times n} \backslash \skup{\mbf{0}}$. Then $\mbf{X}$ is a rank-one, unit Frobenius norm matrix if and only if
 \begin{equation}\label{matrixrank1}
 \mbf{X} \in \mathcal{R}:=\{\mbf{X}: g(\mbf{X})=0 \mbox{ and } f_{ijkl}(\mbf{X})=0 \text{ for all } i<k, j<l  \}.
 \end{equation}
\end{lemma}
\begin{proof}
 If $\mbf{X} \in \R^{m \times n}$ is a rank-one matrix with $\|\mbf{X}\|_F=1$, then by definition there exist two vectors $\mbf{u} \in \R^m$ and $\mbf{v} \in \R^n$ such that $X_{ij}=u_i v_j$ for all $i \in \uglate{m}$, $j \in \uglate{n}$ and $\norma{\mbf{u}}_2=\norma{\mbf{v}}_2=1$.
 Thus
 \begin{align*}
& X_{ij}X_{kl}-X_{il}X_{kj}= u_i v_j u_k v_l - u_i v_l u_k v_j = 0 \\  \text{and} \quad &\sum_{i=1}^m\sum_{j=1}^n X_{ij}^2=\sum_{i=1}^m u_i^2 \sum_{j=1}^n v_j^2=1.
 \end{align*}
For the converse, let $\mbf{X}_{\cdot i}$ represent the $i$-th column of a matrix $\mbf{X} \in \mathcal{R}$. 
Then, for all $j,l \in \uglate{n}$ with $j<l$, it holds
 \begin{equation*}
 X_{ml} \cdot \mbf{X}_{\cdot j} - X_{mj} \cdot \mbf{X}_{\cdot l}= 
 \begin{bmatrix}
 X_{1j}X_{ml} - X_{1l}X_{mj} \\
 X_{2j}X_{ml} -X_{2l}X_{mj} \\
 \vdots \\
 X_{mj}X_{ml} -X_{mj}X_{ml}
 \end{bmatrix}= \mbf{0},
  \end{equation*}
since $X_{ij} X_{ml}=X_{il}X_{mj}$ for all $i \in \uglate{m-1}$ by definition of $\mathcal{R}$. 
 Thus, the columns of the matrix $\mbf{X}$ span a space of dimension one, i.e., the matrix $\mbf{X}$ is a rank-one matrix. From $\sum_{i=1}^m \sum_{j=1}^n X_{ij}^2-1=0$ it follows that the matrix $\mbf{X}$ is  normalized, i.e., $\norma{\mbf{X}}_F=1$.
 \end{proof}
It follows from Lemma \ref{rankone} that the set of rank-one, unit Frobenius norm matrices coincides with
the algebraic variety $\nu_{\R}\zag{J_{M_{mn}}}$ for the ideal $J_{M_{mn}}$ generated by the polynomials $g$ and $f_{ijkl}$, i.e.,
\begin{align}
& J_{M_{m n}} = \interval{\mathcal{G}_{M_{m n}}} \quad \mbox{ with } \nonumber\\
&\mathcal{G}_{M_{m n}}=\skup{g(\mbf{x})} \cup \{f_{ijkl}(\mbf{x}): 1 \leq i < k \leq m, \, 1\leq j < l \leq n \}. \label{J_{M_{m n}}}
\end{align}
Recall that the convex hull of the set $\mathcal{R}$ in \eqref{matrixrank1} forms the unit nuclear norm ball 
and by definition of the theta bodies, 
$$ \overline{\conv\zag{\nu_{\R}\zag{J_{M_{mn}}}}} \subseteq \cdots \subseteq \TB_{k+1}\zag{J_{M_{mn}}} \subseteq \TB_k\zag{J_{M_{mn}}} \subseteq \cdots \subseteq \TB_1\zag{J_{M_{mn}}}.$$ 
Therefore, the theta bodies form closed, convex hierarchical relaxations of the matrix nuclear norm ball. 
In addition, the theta body $\TB_k\zag{J_{M_{mn}}}$ is symmetric, $\TB_k\zag{J_{M_{mn}}} = - \TB_k\zag{J_{M_{mn}}}$. Therefore, it defines 
a unit ball of a norm that we call the \textit{$\theta_k$-norm}.

The next result shows that the generating set of the ideal $J_{M_{mn}}$ introduced above is a Gr{\"o}bner basis.  
  \begin{lemma} \label{Mgroebner}
 The set $\mathcal{G}_{M_{mn}}$ forms the reduced Gr{\"o}bner basis of the ideal $J_{M_{mn}}$ with respect to the grevlex order. 
  \end{lemma}
 \begin{proof}
 The set $\mathcal{G}_{M_{m n}}$ is clearly a basis for the ideal $J_{M_{m n}}$.
 By Proposition \ref{relprime} in the appendix, we only need to check whether the $S$-polynomial, see Definition \ref{Spoly}, 
 satisfies $S\zag{p,q} \rightarrow_{\mathcal{G}_{M_{mn}}} 0$ for all $p,q \in \mathcal{G}_{M_{m n}}$ whenever the leading monomials $\LM\zag{p}$ and $\LM\zag{q}$ are not relatively prime. Here, 
 $S\zag{p,q}  \rightarrow_{\mathcal{G}_{M_{mn}}} 0$ means that $S\zag{p,q}$ reduces to $0$ modulo 
 ${\mathcal{G}_{M_{mn}}}$, see Definition~\ref{def:reducezero}.

Notice that $\LM\zag{g}=x_{11}^2$ and $\LM\zag{f_{ijkl}}=x_{il}x_{kj}$ are relatively prime, for all $1\leq i < k \leq m$ and $1 \leq j < l \leq n$.  
Therefore, we only need to show that $S(f_{ijkl},f_{\hat{i}\hat{j}\hat{k}\hat{l}}) \rightarrow_{\mathcal{G}_{M_{mn}}} 0$ whenever the leading monomials $\LM(f_{ijkl})$ and $\LM(f_{\hat{i}\hat{j}\hat{k}\hat{l}})$ are not relatively prime.
First we consider 
 \begin{equation*}
 f_{ijkl}(\mbf{x})= x_{il} x_{kj} - x_{ij}x_{kl} \quad \text { and } \quad  f_{i\hat{j}\hat{k}l}(\mbf{x})= x_{il} x_{\hat{k}\hat{j}} - x_{i\hat{j}}x_{\hat{k}l}
 \end{equation*}
 for $1 \leq i < k < \hat{k} \leq m, \, 1 \leq j < \hat{j} <l  \leq n$.
The $S$-polynomial 
is then of the form
\begin{align*}
S(f_{ijkl},f_{i\hat{j}\hat{k}l}) =  x_{\hat{k}\hat{j}}f_{ijkl}(\mbf{x}) - x_{kj}f_{i\hat{j}\hat{k}l}(\mbf{x})&= - x_{ij}x_{kl}x_{\hat{k}\hat{j}} + x_{i\hat{j}}x_{\hat{k}l}x_{kj}\\&=x_{\hat{k}l} f_{ijk\hat{j}}(\mbf{x}) - x_{ij} f_{k\hat{j}\hat{k}l}(\mbf{x}) \in J_{M_{mn}}
\end{align*}
so that $S(f_{ijkl},f_{i\hat{j}\hat{k}l})  \rightarrow_{\mathcal{G}_{M_{mn}}} 0$.
The remaining cases are treated with similar arguments. 

In order to show that $\mathcal{G}_{M_{mn}}$ is a reduced Gr{\"o}bner basis (see Definition \ref{reducedGr}), 
we first notice that $\LC(f)=1$ for all $f \in \mathcal{G}_{M_{m n}}$. In addition, the leading monomial of $f \in \mathcal{G}_{M_{mn}}$ is always of degree two and there are no two different polynomials $f_i,f_j \in \mathcal{G}_{M_{m n}}$ such that $\LM(f_i)=\LM(f_j)$. Therefore, $\mathcal{G}_{M_{mn}}$ is the reduced Gr{\"o}bner basis of the ideal $J_{M_{mn}}$ with respect to the grevlex order.
\end{proof}

The Gr{\"o}bner basis  $\mathcal{G}_{M_{m n}}$ of $J_{M_{mn}}=\interval{\mathcal{G}_{M_{mn}}}$ yields the $\theta$-basis of $\R[\mbf{x}]/J_{M_{mn}}$.
For the sake of simplicity, we only provide its elements up to degree two,
\begin{align*}
\mathcal{B}_1&=\skup{1+J_{M_{mn}}, x_{11}+J_{M_{mn}}, x_{12}+J_{M_{mn}},  \ldots, x_{mn}+J_{M_{mn}} } \\
\mathcal{B}_2&=\mathcal{B}_1 \cup \skup{x_{ij}x_{kl}+J_{M_{mn}}: \zag{i,j,k,l} \in \mathcal{S}_{\mathcal{B}_2}},
\end{align*}
where $\mathcal{S}_{\mathcal{B}_2}=\skup{\zag{i,j,k,l}: 1 \leq i \leq k \leq m, 1 \leq j \leq l \leq n} \backslash \zag{1,1,1,1}$. 
Given the $\theta$-basis, the theta body $\TB_k(J_{M_{mn}})$ is well-defined. We formally introduce an associated norm next.
\begin{definition} \label{def:theta}
The matrix \textit{$\theta_k$-norm}, denoted by $\norma{\cdot}_{\theta_k}$, is the norm induced by the $k$-theta body $\TB_k\zag{J_{M_{mn}}}$, i.e.,
\begin{equation*}
\norma{\mbf{X}}_{\theta_k}=\inf \skup{r: \mbf{X} \in r\TB_k\zag{J_{M_{mn}}}}.
\end{equation*}
\end{definition}
The $\theta_k$-norm can be computed with the help of Theorem~\ref{closureQ}, i.e., as 
\[
\norma{\mbf{X}}_{\theta_k} = \min t \quad \mbox{ subject to } \mbf{X} \in t \mbf{Q}_{\mathcal{B}_k}(J_{M_{mn}}).
\]
Given the moment matrix $\mbf{M}_{\mathcal{B}_k}[\mbf{y}]$ associated with $J_{M_{mn}}$, this minimization program is equivalent to the semidefinite program
\begin{equation}\label{matrix:sdp}
\min_{t \in \R , \mbf{y} \in \R^{\mathcal{B}_k}} t \quad \mbox{ subject to } \quad  \mbf{M}_{\mathcal{B}_k}[\mbf{y}] \succcurlyeq 0, y_0 = t, \mbf{y}_{\mathcal{B}_1} = \mbf{X}.
\end{equation}
The last constraint might require some explanation. The vector $\mbf{y}_{\mathcal{B}_1}$ denotes the restriction
of $\mbf{y}$ to the indices in $\mathcal{B}_1$, where the latter can be identified with the set $[m] \times [n]$ indexing the matrix entries.
Therefore, $\mbf{y}_{\mathcal{B}_1} = \mbf{X}$ means componentwise $y_{x_{11}+J} = X_{11}, y_{x_{12}+J} = X_{12}, \hdots, y_{x_{mn} + J} = X_{mn}$.
For the purpose of illustration, 
we focus on the $\theta_1$-norm in $\R^{2 \times 2}$ in 
Section~\ref{matrix2x2} below, and provide a step-by-step procedure for building the corresponding semidefinite program in \eqref{matrix:sdp}.

\ifnocut
Notice that the number of elements in $\mathcal{B}_1$ is $mn+1$, and in $\mathcal{B}_2 \backslash \mathcal{B}_1$ is $\frac{m\cdot(m+1)}{2} \cdot \frac{n \cdot (n+1)}{2} -1\sim \frac{\zag{mn}^2}{2}$, i.e., the number of elements of the $\theta$-basis restricted to the degree $2$ scales polynomially in the total number of matrix entries $mn$. 
Therefore, the computational complexity of the SDP in \eqref{matrix:sdp} is polynomial in $mn$.
\fi

We will show next that the theta body $\TB_1(J)$ and hence, all $\TB_k(J)$ for $k \in \N$, coincide with the nuclear norm ball.
To this end, the following lemma provides expressions for the boundary of the matrix nuclear unit norm ball.

\begin{lemma}\label{trnucn<p}
Let $\mathcal{O}_{c}$ ($\mathcal{O}_r$) denote the set of all matrices $\mbf{M} \in \R^{n \times m}$ with orthonormal columns (rows), i.e., $\mathcal{O}_c=\skup{\mbf{M} \in \R^{n \times m}: \mbf{M}^T\mbf{M}=\mbf{I}_m}$ and $\mathcal{O}_r=\skup{\mbf{M} \in \R^{n \times m}: \mbf{M}\mbf{M}^T=\mbf{I}_n}$. Then
\begin{equation}\label{nuclear}
\skup{\mbf{X} \in \R^{m \times n}: \norma{\mbf{X}}_*\leq 1}=\skup{\mbf{X} \in \R^{m \times n}:\tr\zag{\mbf{M}\mbf{X}} \leq 1, \text{ for all } \mbf{M} \in \mathcal{O}_{c} \cup \mathcal{O}_{r}}.
\end{equation}
\end{lemma}

\begin{remark}
Notice that $\mathcal{O}_c=\emptyset$ for $m>n$ and $\mathcal{O}_r=\emptyset$ for $m<n$.
\end{remark}

\begin{proof}
If suffices to treat the case $m \leq n$ because $\norma{\mbf{X}}_*=\norma{\mbf{X}^T}_*$ for all matrices $\mbf{X}$, 
and $\mbf{M} \in \mathcal{O}_r$ if and only if $\mbf{M}^T \in \mathcal{O}_c$. 
Let $\mbf{X} \in \R^{m \times n}$ such that $\norma{\mbf{X}}_*\leq 1$ and let $\mbf{X}=\mbf{U}\mbf{\Sigma}\mbf{V}^T$ be its singular value decomposition. 
For $\mbf{M} \in \mathcal{O}_{c}$, the spectral norm satisfies $\|\mbf{M}\| \leq 1$ and therefore, using that the nuclear norm is the dual
of the spectral norm, see e.g.~\cite[p.~96]{bh96}, 
\begin{align*}
\tr\zag{\mbf{M}\mbf{X}} \leq \|\mbf{M}\| \cdot \|\mbf{X}\|_* \leq \norma{\mbf{X}}_* \leq 1.
\end{align*}
For the converse, let $\mbf{X} \in \R^{m \times n}$ be such that $\tr\zag{\mbf{M}\mbf{X}} \leq 1$, for all $\mbf{M} \in \mathcal{O}_c$. Let  $\mbf{X}=\mbf{U}\overline{\mbf{\Sigma}}\,\overline{\mbf{V}}^T$ denote its reduced singular value decomposition, i.e., $\mbf{U},\overline{\mbf{\Sigma}} \in \R^{m \times m}$ and $\overline{\mbf{V}} \in \R^{n \times m}$ with $\mbf{U}^T\mbf{U}=\mbf{U}\mbf{U}^T=\overline{\mbf{V}}^T\overline{\mbf{V}}=\mbf{I}_m$. 
Since $\mbf{M}:=\overline{\mbf{V}}\mbf{U}^T \in \mathcal{O}_c$, it follows that 
$$ 1\geq \tr(\mbf{MX})=\tr(\overline{\mbf{V}}\mbf{U}^T\mbf{U}\overline{\mbf{\Sigma}}\,\overline{\mbf{V}}^T)=\tr(\overline{\mbf{\Sigma}})=\norma{\mbf{X}}_*.$$
This completes the proof.
\end{proof}
Next, using Lemma \ref{trnucn<p}, we show that the theta body $\TB_1(J)$ equals the nuclear norm ball.
This result is related to Theorem 4.4 in \cite{grande2014theta}.
\begin{theorem}\label{TH1n<p}
The polynomial ideal $J_{M_{m n}}$ defined in \eqref{J_{M_{m n}}} is $\TB_1$-exact, i.e.,
 $$\TB_1\zag{J_{M_{m  n}}}=\conv\zag{\mbf{x} : g(\mbf{x})=0, f_{ijkl}(\mbf{x})=0 \text{ for all }i<k, j<l}.$$
 In other words,
 $$
 \skup{\mbf{X} \in \R^{m \times n}: \mbf{X} \in \TB_1\zag{J_{M_{mn}}}}=\skup{\mbf{X} \in \R^{m \times n}: \norma{\mbf{X}}_* \leq 1}.
 $$
\end{theorem}

\begin{proof}
By definition of $\TB_1(J_{M_{mn}})$, it is enough to show that the boundary of the unit nuclear norm can be written as $1$-sos mod $J_{M_{mn}}$, 
which by Lemma \ref{trnucn<p} means that the polynomial $1-\sum_{i=1}^m\sum_{j=1}^n{x_{ij}M_{ji}}$ is $1$-sos mod $J_{M_{mn}}$ for all $\mbf{M} \in \mathcal{O}_c \cup \mathcal{O}_r$. We start by fixing $\mbf{M}=\begin{pmatrix} \mbf{I}_m \\ \mbf{0} \end{pmatrix}$ in case $m\leq n$ and $\mbf{M}=\begin{pmatrix} \mbf{I}_n & \mbf{0} \end{pmatrix}$ in case $m > n$, where $\mbf{I}_k \in \R^{k \times k}$ is the identity matrix. For this choice of $\mbf{M}$,
we need to show that $1-\sum_{i=1}^{\ell} x_{ii}$ is $1$-sos mod $J_{M_{mn}}$, where $\ell=\min \skup{m,n}$.
Note that
\begin{align*}
1-\sum_{i=1}^{\ell}x_{ii}=&\frac{1}{2}\left[\zag{1-\sum_{i=1}^{\ell} x_{ii}}^2 + \zag{1-\sum_{i=1}^m\sum_{j=1}^n x_{ij}^2}+\sum_{i<j\leq {\ell}}\zag{x_{ij}-x_{ji}}^2 \right.\\ & \left. -2\sum_{i<j\leq {\ell}}\zag{x_{ii}x_{jj}-x_{ij}x_{ji}} +\sum_{i=1}^m\sum_{j=m+1}^n x_{ij}^2 +  \sum_{i=n+1}^m\sum_{j=1}^n x_{ij}^2\right],
\end{align*}
\ifnocut
since
\begin{align*}
\zag{1-\sum_{i=1}^{\ell} x_{ii}}^2=&1-2\sum_{i=1}^{\ell} x_{ii}+\sum_{i=1}^{\ell}\sum_{j=1}^{\ell} x_{ii}x_{jj} \\ &= 1-2\sum_{i=1}^{\ell} x_{ii}+2\sum_{i<j\leq {\ell}}x_{ii}x_{jj} +\sum_{i=1}^{\ell} x_{ii}^2,
\end{align*}
\begin{align*}
1-\sum_{i=1}^m\sum_{j=1}^n x_{ij}^2+\sum_{i=1}^m\sum_{j=m+1}^n x_{ij}^2 \,+ &\sum_{i=n+1}^m \sum_{j=1}^n x_{ij}^2  =1-\sum_{i=1}^{\ell} \sum_{j=1}^{\ell} x_{ij}^2 \\ &=1-\sum_{i<j\leq {\ell}}\zag{x_{ij}^2+x_{ji}^2}-\sum_{i=1}^{\ell} x_{ii}^2,
\end{align*}
and
\begin{align*}
\sum_{i<j\leq {\ell}}\zag{x_{ij}-x_{ji}}^2-\,& 2\sum_{i<j\leq {\ell}}\zag{x_{ii}x_{jj}-x_{ij}x_{ji}}\\=&\sum_{i<j\leq {\ell}} \zag{x_{ij}^2+x_{ji}^2-2x_{ij}x_{ji}-2 x_{ii}x_{jj}+2 x_{ij}x_{ji}}
\\
=&\sum_{i<j\leq {\ell}} \zag{x_{ij}^2+x_{ji}^2}-2\sum_{i<j \leq {\ell}}x_{ii}x_{jj}.
\end{align*}
\fi
Therefore, $1-\sum_{i=1}^{\ell} x_{ii}$ is $1$-sos mod $J_{M_{mn}}$, since the polynomials $1-\sum_{i=1}^{\ell} x_{ii}$, $x_{ij}-x_{ji}$, $x_{ij}$, and $x_{ji}$ are linear  
and the polynomials $1-\sum_{i=1}^m\sum_{j=1}^n x_{ij}^2$ and $2\zag{x_{ii}x_{jj}-x_{ij}x_{ji}}$ are contained in the ideal, for all $i<j\leq {\ell}$.

Next, we define transformed variables
$$x'_{ij}=
\begin{cases} 
\sum_{k=1}^m M_{ik} x_{kj} & \mbox{if } m \leq n, \\
\sum_{k=1}^n x_{ik} M_{kj} & \mbox{if } m> n. 
\end{cases}
$$
Since  $x'_{ij}$ is a linear combination of $\{ x_{kj}\}_{k=1}^{m} \cup \{ x_{ik}\}_{k=1}^{n}$, for every $i \in \uglate{m}$ and $j \in \uglate{n}$,
linearity of the polynomials $1-\sum_{i=1}^{\ell} x'_{ii}$, $x'_{ij}-x'_{ji}$,  $x'_{ij}$, and $x'_{ji}$ is preserved, for all $i<j$. It remains to show that the ideal is invariant under this transformation. For the polynomial $1-\sum_{i=1}^m\sum_{j=1}^n {x'_{ij}}^2$ this is clear since $\mbf{M} \in \R^{n \times m}$ has unitary columns in case when $m\leq n$ and unitary rows in case $m\geq n$. In the case of $m\leq n$ the polynomial $x'_{ii}x'_{jj}-x'_{ij}x'_{ji}$ is contained in the ideal $J$ since 
$$
x'_{ii}x'_{jj}-x'_{ij}x'_{ji}=\sum_{k=1}^m\sum_{l=1}^m M_{ik}M_{jl}\zag{x_{ki} x_{lj}-x_{kj}x_{li}} 
$$
and the polynomials $x_{ki} x_{lj}-x_{kj}x_{li}$ are contained in $J$ for all $i<j\leq m$.
Similarly, in case $m\geq n$ the polynomial $x'_{ii}x'_{jj}-x'_{ij}x'_{ji}$ is in the ideal since 
$$x'_{ii}x'_{jj}-x'_{ij}x'_{ji}=\sum_{k=1}^n\sum_{l=1}^n M_{ki}M_{lj}\zag{x_{ik} x_{jl}-x_{il}x_{jk}} $$
and  polynomials $x_{ik} x_{jl}-x_{il}x_{jk}$ are in the ideal, for all $i<j\leq n$.
\end{proof}

The following corollary is a direct consequence of Theorem \ref{TH1n<p} and the nestedness property \eqref{seqtheta} of theta bodies.
\begin{corollary}
The matrix $\theta_1$-norm coincides with the matrix nuclear norm, i.e.,
\begin{equation*}
\norma{\mbf{X}}_*=\norma{\mbf{X}}_{\theta_1}, \quad \text{ for all } \mbf{X} \in \R^{m \times n}.
\end{equation*}
Moreover,
$$ \TB_1\zag{J_{M_{m n}}} = \TB_2\zag{J_{M_{m n}}} = \cdots = \conv\zag{\nu_{\R}\zag{J_{M_{m n}}}}.$$
\end{corollary}

\ifnocut
\begin{remark} The ideal \eqref{J_{M_{m n}}} is not the only choice that satisfies \eqref{variety:rankone}.
For example, in \cite{chparewi10} the following polynomial ideal was suggested 
\begin{equation}
J=\interval{\skup{x_{ij}-u_iv_j}_{i \in \uglate{m},j \in \uglate{n}}, \sum_{i=1}^m u_i^2-1, \sum_{j=1}^n v_j^2-1}
\end{equation}
in $\R\uglate{\mbf{x}, \mbf{u},\mbf{v}}=\R\uglate{x_{11},\ldots,x_{mn},u_1,\ldots,u_m,v_1,\ldots,u_n}$. 
Some tedious computations reveal the reduced Gr{\"o}bner basis $\mathcal{G}$ of the ideal $J$  with respect to the grevlex (and grlex) ordering, 
\begin{align}
\mathcal{G}=&\left\{ g_1^{ij}= x_{ij}-u_iv_j :  
i\in\uglate{m},\, j\in\uglate{n}\right\}  
\bigcup  \left\{ g_2 =\sum_{i=1}^m u_i^2 -1, g_3 = \sum_{j=1}^n v_j^2 -1\right\} \nonumber\\
\bigcup & \left\{ g_4^{i,j,k} =  x_{ij}u_k-x_{kj}u_i : 1\leq i <k \leq m, \, j \in \uglate{n}\right\}  \nonumber\\
\bigcup &  \left\{ g_5^{i,j,k}  = x_{ij}v_k-x_{ik}v_j : i \in \uglate{m} ,\, 1\leq j <k \leq n \right\} \nonumber\\
\bigcup & \left\{ g_6^i =  \sum_{j=1}^n x_{ij}v_j-u_i  : i \in \uglate{m}\right\}
\bigcup  \left\{ g_7^j = \sum_{i=1}^m x_{ij}u_i-v_j: j \in \uglate{n} \right\}  \nonumber\\
\bigcup & \left\{ g_8^{i,j} =  \sum_{k=1}^n x_{ik}x_{jk}-u_i u_j : 1 \leq i < j \leq m\right\} \nonumber\\
\bigcup & \left\{ g_9^{i,j} = \sum_{k=1}^m x_{ki} x_{kj}-v_i v_j :   1 \leq i < j \leq n \right\} \nonumber\\
\bigcup & \left\{ g_{10}^i =  \sum_{j=1}^n x_{ij}^2-u_i^2  : 2 \leq i \leq m \right\} 
\bigcup  \left\{ g_{11}^j =  \sum_{i=1}^m x_{ij}^2-v_j^2  : 2 \leq j \leq n \right\} \nonumber\\
\bigcup & \left\{ g_{12}^{i,j,k,l} =  x_{ij}x_{kl}- x_{il}x_{kj} :  1\leq i <k \leq m, \, 1\leq j <l \leq n \right\} \nonumber\\
\bigcup &\left\{ g_{13} =  x_{11}^2-\sum_{i=2}^{m}\sum_{j=2}^{n}x_{ij}^2+ \sum_{i=2}^{m} u_i^2 + \sum_{j=2}^{n} v_j^2-1 \right\}.
\label{groebner:matrix}
\end{align}
Obviously, this Gr{\"o}bner basis is much more complicated than the one of the ideal $J_{M_{mn}}$ introduced above. Therefore, computations (both theoretical and numerical) with this alternative ideal seem to be more demanding.
In any case, the variables $\skup{u_i}_{i=1}^m$ and $\skup{v_j}_{j=1}^n$ are only auxiliary ones, so one would like to eliminate these from the above Gr{\"o}bner basis. 
By doing so, one obtains the Gr{\"o}bner basis $\mathcal{G}_{M_{mn}}$ defined in \eqref{J_{M_{m n}}}. Notice that $\sum_{i=1}^m\sum_{j=1}^n x_{ij}^2-1=g_{13}+\sum_{i=2}^m g_{10}^i + \sum_{j=2}^n g_{11}^j$ together with $\{g_{12}^{i,j,k,l}\}$ form the basis $\mathcal{G}_{M_{m n}}$.
\end{remark}
\fi

\subsection{The $\theta_1$-norm in $\R^{2 \times 2}$} \label{matrix2x2}

For the sake of illustration, we consider the specific example of $2 \times 2$ matrices and provide the corresponding
semidefinite program for the computation of the $\theta_1$-norm explicitly. 
Let us denote the corresponding polynomial ideal in $\R\uglate{\mbf{x}}=\R\uglate{x_{11},x_{12},x_{21},x_{22}}$ simply by 
\begin{equation}\label{ideal}
J = J_{M_{22}} = \interval{x_{12}x_{21}-x_{11}x_{22}, x_{11}^2+x_{12}^2+x_{21}^2+x_{22}^2-1}
\end{equation}
The associated algebraic variety is of the form
\begin{equation*}
\nu_{\R}\zag{J}=\skup{\mbf{x}: x_{12}x_{21}=x_{11}x_{22},x_{11}^2+x_{12}^2+x_{21}^2+x_{22}^2=1}
\end{equation*}
and corresponds to the set of rank-one matrices with $\|\mbf{X}\|_F = 1$. 
Its convex hull  
consists of matrices 
$\mbf{X} \in \R^{2 \times 2}$ with $\|\mbf{X}\|_* \leq 1$. 
According to Lemma \ref{Mgroebner}, the Gr{\"o}bner basis $\mathcal{G}$ of $J$ with respect to the grevlex 
order is 
\begin{equation*}
\mathcal{G}=\skup{g_1 = x_{12}x_{21}-x_{11}x_{22} , \, g_2 = x_{11}^2+x_{12}^2+x_{21}^2+x_{22}^2-1}
\end{equation*}
with the corresponding $\theta$-basis $\mathcal{B}$ of $\R\uglate{\mbf{x}}/J$ restricted to the degree two given as
\begin{align*}
\mathcal{B}_1&=\skup{1+J, x_{11}+J, x_{12}+J,  x_{21}+J, x_{22}+J } \\
\mathcal{B}_2&=\mathcal{B}_1 \cup \{x_{11}x_{12}+J, x_{11}x_{21}+J, x_{11}x_{22}+J,  x_{12}^2+J,  x_{12}x_{22}+J, \\ & \quad \quad   x_{21}^2+J,  x_{21}x_{22}+J, x_{22}^2+J \}.
\end{align*}
The set $\mathcal{B}_2$ consists of all monomials of degree at most two which are not divisible by a leading term of any of the polynomials inside the Gr{\"o}bner basis $\mathcal{G}$. For example, $x_{11}x_{12}+J$ is an element of the theta basis $\mathcal{B}$, but $x_{11}^2+J$ is not since $x_{11}^2$ is divisible by $\LT(g_2)$.

Linearizing the elements of $\mathcal{B}_2$ results in Table \ref{linearM}, where the monomials $f$ in the first row stand for an element $f+J \in \mathcal{B}_2$.
\begin{table}
\caption{Linearization of the elements of $\mathcal{B}_2=\{f+J\}$ for matrix $2 \times 2$ case.}
\label{linearM}
\begin{tabular}{c c c c c c c }
\hline\noalign{\smallskip}
$1+J$ & $x_{11}+J$ & $x_{12}+J$ & $x_{21}+J$ & $x_{22}+J$ & $ x_{11}x_{12}+J$ & $x_{11}x_{21}+J$ \\
$y_0$& $x_{11}$ & $x_{12}$ & $x_{21}$ & $x_{22}$ & $y_1$ & $y_2$ \\
\noalign{\smallskip}\hline\vspace*{-0.3cm}
\end{tabular}
\begin{tabular}{c c c c c c }
\hline\noalign{\smallskip}
 $x_{11}x_{22}+J $ & $x_{12}^2+J$ & $x_{12}x_{22}+J$  & $x_{21}^2+J$ & $x_{21}x_{22}+J$ & $x_{22}^2+J$  \\
$y_3$ & $y_4$ & $y_5$ & $y_6$ & $y_7$ & $y_8$   \\
\noalign{\smallskip}\hline
\end{tabular}
\end{table}
Therefore, $\uglate{\mbf{x}}_{\mathcal{B}_1}=\zag{1, x_{11}, x_{12}, x_{21}, x_{22}}^T$ and 
the following combinatorial moment matrix $\mbf{M}_{\mathcal{B}_1}\zag{\mbf{x},\mbf{y}}$, see Definition~\ref{def:moment}, is given as
$$
\mbf{M}_{\mathcal{B}_1}\zag{\mbf{x},\mbf{y}}=
\begin{bmatrix}
y_0        & x_{11}                          & x_{12} & x_{21} & x_{22} \\
x_{11}  & -y_{4}-y_{6}-y_{8}+y_0 & y_{1}    & y_{2}  & y_{3}  \\
x_{12}  &  y_1                               & y_{4}    & y_{3}  & y_{5} \\
x_{21}  &  y_2                               & y_3       & y_{6}  & y_{7}  \\
x_{22}  &  y_3                               &  y_5      &  y_7    & y_{8} 
\end{bmatrix}.
$$
For instance, the entry $(2,2)$ of  $\uglate{\mbf{x}}_{\mathcal{B}_1}\uglate{\mbf{x}}_{\mathcal{B}_1}^T$ is of the form $x_{11}^2+J = -x_{12}^2-x_{21}^2-x_{22}^2+1+J$, where 
we exploit the second property in Definition~\ref{thetabasis} and the fact that $g_2 \in J$.
Replacing $x_{12}^2+J$ by $y_4$, etc.\  as in Table \ref{linearM},  yields the stated expression for $\mbf{M}_{\mathcal{B}_1}\zag{\mbf{x},\mbf{y}}_{2,2}$.

By Theorem \ref{closureQ}, the first theta body $\TB_1\zag{J}$ is the closure of
\begin{equation*}
\mbf{Q}_{\mathcal{B}_1}\zag{J}=\pi_{\mbf{x}}\skup{\zag{\mbf{x},\mbf{y}} \in \R^{\mathcal{B}_{2}}: \mbf{M}_{\mathcal{B}_1}\zag{\mbf{x},\mbf{y}} \succeq 0, \, y_0=1},
\end{equation*}
where $\pi_{\mbf{x}}$ represents the projection onto the variables $\mbf{x}$, i.e., the projection onto $x_{11}$, $x_{12}$, $x_{21}$, $x_{22}$.
Furthermore, $\theta_1$-norm of a matrix $\mbf{X} \in \R^{2 \times 2}$ induced by the $\TB_1\zag{J}$ and denoted as $\norma{\cdot}_{\theta_1}$ can be computed as
\begin{equation}
\norma{\mbf{X}}_{\theta_1}=\inf t \text{ s.t. } \mbf{X} \in t\mbf{Q}_{\mathcal{B}_1}\zag{J}
\end{equation}
which is equivalent to

\begin{equation}\label{MatrixM1}
 \inf_{t \in \R,\mbf{y} \in \R^8} t   \quad \text{ s.t. } \quad
\mbf{M}= \begin{bmatrix}
t & X_{11} & X_{12} & X_{21} & X_{22} \\
X_{11} & -y_{4}-y_{6}-y_{8}+t & y_{1} & y_{2} & y_{3}  \\
X_{12} & y_1  & y_{4} & y_{3} & y_{5} \\
X_{21} & y_2  & y_3   & y_{6} & y_7 \\
X_{22} & y_3 &  y_5  &  y_7   & y_{8} \\
 \end{bmatrix}
 \succeq 0.
\end{equation}

Notice that $\trace(\mbf{M})=2t$.
By Theorem \ref{TH1n<p}, the above  program is equivalent to the standard semidefinite program for computing the nuclear norm of a given matrix $\mbf{X} \in \R^{m \times n}$ 
$$
 \min_{\mbf{W},\mbf{Z}} \frac{1}{2}\zag{\trace\zag{\mbf{W}}+\trace\zag{\mbf{Z}}}  \quad \text{ s.t. } \quad
 \begin{bmatrix}
W_{11} & W_{12} & X_{11} & X_{12}  \\
W_{12} & W_{22} & X_{21} & X_{22}  \\
X_{11} & X_{21}  & Z_{11}  & Z_{12} \\
X_{22} & X_{22}  & Z_{12}  & Z_{22}  \\
 \end{bmatrix}
 \succeq 0.
$$

\begin{remark}
In compressive sensing, reconstruction of sparse signals via $\ell_1$-norm minimization  is well-understood, see for example \cite{carota06,do06-2,fora13}. 
It is possible to provide hierarchical relaxations via theta bodies of the unit $\ell_1$-norm ball. However, as in the matrix scenario discussed above, all these relaxations coincide with the unit $\ell_1$-norm ball, \cite{ThesisStojanac}.
\end{remark}
\ifnocut

\fi

\section{The tensor $\theta_k$-norm}
\label{sec:tensorTH}

Let us now turn to the tensor case and study the hierarchical closed convex relaxations of the unit tensor nuclear norm ball defined via theta bodies.
Since in the matrix case all $\theta_k$-norms are equal to the matrix nuclear norm, their generalization to the tensor case may all be viewed as natural 
generalizations of the nuclear norm. We focus mostly on the $\theta_1$-norm whose unit norm ball is the largest in a hierarchical sequence of relaxations.  
Unlike in the matrix case, the $\theta_1$-norm defines a new tensor norm, that up to the best of our knowledge has not been studied before.

The polynomial ideal will be generated by the minors of order two of the unfoldings -- and matricizations in the case $d \geq 4$ -- of the tensors, where each
variable corresponds to one entry in the tensor. As we will see, a tensor is of rank one if and only if all order-two minors of 
the unfoldings (matricizations) vanish. While the order-three case requires to consider all three unfoldings, there are several possibilities
for the order-$d$ case when $d \geq 4$. In fact, a $d$th-order tensor is of rank one if all minors of all unfoldings vanish so that it may be enough
to consider only the unfoldings. However, one may as well consider the ideal generated by all minors of {\it all matricizations} or one
may consider a subset of matricizations including all unfoldings. Indeed, any {\it tensor format} -- and thereby any notion of tensor rank -- 
corresponds to a set of matricizations and in this 
way, one may associate a $\theta_k$-norm to a certain tensor format. We refer to e.g.~\cite{ha12-4,rascst15} for some background on various tensor formats. However, as we will show later, the corresponding reduced Gr{\"o}bner basis with respect to the grevlex order does not depend on the choice of the tensor format.
We will mainly concentrate on the case that {\it all} matricizations are taken into account for defining the ideal. Only for the case $d=4$, we will briefly
discuss the case, that the ideal is generated only by the minors corresponding to the four unfoldings.

Below, we consider first the special case of third-order tensors and continue then with fourth-order tensors. 
In Subsection~\ref{dtensor} we will treat the general $d$th-order case. 

\subsection{Third-order tensors}

As described above, we will consider the order-two minors of all the unfoldings of a third-order tensor.
Our notation requires the following sets of subscripts
\begin{align*}
\mathcal{S}_1 & = \skup{\zag{\bm{\alpha},\bm{\beta}} : 1 \leq \alpha_1 < \beta_1 \leq n_1,\, 1 \leq \beta_2 < \alpha_2 \leq n_2,\, 1 \leq \beta_3 \leq \alpha_3 \leq n_3}, \\
\mathcal{S}_2 & = \skup{\zag{\bm{\alpha},\bm{\beta}} : 1 \leq \alpha_1 \leq \beta_1 \leq n_1,\, 1 \leq \beta_2 < \alpha_2 \leq n_2,\, 1 \leq \alpha_3 < \beta_3 \leq n_3}, \\
\mathcal{S}_3 & = \skup{ \zag{\bm{\alpha},\bm{\beta}} : 1 \leq \alpha_1 < \beta_1 \leq n_1, \,1 \leq \alpha_2 \leq \beta_2 \leq n_2,\, 1 \leq \beta_3 < \alpha_3 \leq n_3}, \\
\overline{\mathcal{S}}_i &=\skup{\zag{\bm{\alpha},\bm{\beta}}: \zag{\bm{\alpha},\bm{\beta}} \in \mathcal{S}_i \text{ and } \alpha_j\neq\beta_j,  \text{ for all } j \in \uglate{3}}, \quad \text{for all } i \in \uglate{3}.
\end{align*}
The following polynomials $f^{\zag{\bm{\alpha},\bm{\beta}}}$ in $\R\uglate{\mbf{x}}=\R\uglate{x_{111},x_{112},\ldots,x_{n_1n_2n_3}}$ correspond to a subset of all order-two minors of all tensor unfoldings,
\begin{align*}
f^{\zag{\bm{\alpha}, \bm{\beta}}}(\mbf{x})&=x_{\bm{\alpha}} x_{\bm{\beta}}-x_{\bm{\alpha} \vee \bm{\beta}} x_{\bm{\alpha} \wedge \bm{\beta}}, \quad \zag{\bm{\alpha},\bm{\beta}} \in \mathcal{S}:=\mathcal{S}_1 \cup \mathcal{S}_2 \cup \mathcal{S}_3 \\g_3(\mbf{x})&= \sum_{i=1}^{n_1} \sum_{j=1}^{n_2} \sum_{k=1}^{n_3}  x_{ijk}^2-1, 
\end{align*}
where $\uglate{\bm{\alpha} \vee\bm{\beta}}_i=\max\skup{\alpha_i,\beta_i}$ and $\uglate{\bm{\alpha} \wedge \bm{\beta}}_i=\min\skup{\alpha_i,\beta_i}$. In particular, the following order-two minor of $\mbf{X}^{\{1\}}$ is not contained in $\skup{f^{\zag{\bm{\alpha},\bm{\beta}}}: \zag{\bm{\alpha},\bm{\beta}} \in \mathcal{S}}$
$$f=x_{\bm{\alpha}}x_{\bm{\beta}}-x_{\hat{\bm{\alpha}}}x_{\hat{\bm{\beta}}}, \quad \text{where } \hat{\bm{\alpha}}=\zag{\alpha_1,\beta_2,\beta_3}, \hat{\bm{\beta}}=\zag{\beta_1,\alpha_2,\alpha_3} \text{ and } \zag{\bm{\alpha},\bm{\beta}} \in \overline{\mathcal{S}}_3. $$ 

 We remark that in real algebraic geometry and commutative algebra, polynomials $f^{\zag{\bm{\alpha},\bm{\beta}}}$ are known as Hibi relations, see \cite{hibi065413015}. 
\begin{lemma}\label{TH:3rdrankone}
A tensor $\mbf{X} \in \R^{n_1 \times n_2 \times n_3}$ is a rank-one, unit Frobenius norm tensor if and only if
\begin{equation}\label{TH:cond3rd}
g_3(\mbf{X})=0 \, \text{ and }\,f^{\zag{\bm{\alpha},\bm{\beta}}}(\mbf{X})=0 \quad \text{for all} \quad \zag{\bm{\alpha},\bm{\beta}} \in \mathcal{S}.
\end{equation}
\end{lemma}

\begin{proof}
Sufficiency of \eqref{TH:cond3rd} follows directly from the definition of the rank-one unit Frobenius norm tensors.
For necessity, the first step is to show that mode-$1$ fibers (columns)  span one-dimensional space in $\R^{n_1}$.
To this end, we note that for $\beta_2 \leq \alpha_2$ and $\beta_3 \leq \alpha_3$, the fibers $\mbf{X}_{\cdot \alpha_2\alpha_3}$ and $\mbf{X}_{\cdot {\beta}_2 {\beta}_3}$ satisfy
\begin{align*}
-X_{n_1\alpha_2\alpha_3} 
\begin{bmatrix}
 X_{1\beta_2\beta_3} \\
 X_{2\beta_2\beta_3} \\
 \vdots \\
  X_{n_1\beta_2\beta_3}
\end{bmatrix}&+
X_{n_1\beta_2\beta_3} 
\begin{bmatrix}
 X_{1\alpha_2\alpha_3} \\
 X_{2\alpha_2\alpha_3} \\
 \vdots \\
 X_{n_1\alpha_2\alpha_3}
\end{bmatrix}\\
&=
\begin{bmatrix}
-X_{1\beta_2\beta_3}X_{n_1\alpha_2\alpha_3} + X_{1\beta_2\beta_3}X_{n_1\alpha_2\alpha_3} \\
-X_{2\beta_2\beta_3}X_{n_1\alpha_2\alpha_3} + X_{2\beta_2\beta_3}X_{n_1\alpha_2\alpha_3} \\
\vdots \\
-X_{n_1\beta_2\beta_3}X_{n_1\alpha_2\alpha_3} + X_{n_1\beta_2 \beta_3}X_{n_1\alpha_{2}\alpha_3}
\end{bmatrix}=\mbf{0},
\end{align*}
where we used that $f^{\zag{\bm{\alpha},\bm{\beta}}}(\mbf{X})=0$ for all $\zag{\bm{\alpha},\bm{\beta}} \in \mathcal{S}$.
From $g_3\zag{\mbf{X}}=0$ it follows that the tensor $\mbf{X}$ is normalized.

Using similar arguments, one argues that mode-$2$ fibers (rows) and mode-$3$ fibers span one dimensional spaces in $\R^{n_2}$ and $\R^{n_3}$, respectively.
This completes the proof.
\end{proof}

A third-order tensor $\mbf{X} \in \R^{n_1 \times n_2 \times n_3}$ is rank one if and only if all three unfoldings $\mbf{X}^{\{1\}} \in \R^{n_1 \times n_2 n_3}$, $\mbf{X}^{\{2\}} \in \R^{n_2 \times n_1 n_3}$, and $\mbf{X}^{\{3\}} \in \R^{n_3 \times n_1 n_2}$ are rank-one matrices. Notice that $f^{\zag{\bm{\alpha},\bm{\beta}}}(\mbf{X})=0$ for all $ \zag{\bm{\alpha},\bm{\beta}} \in \mathcal{S}_{\ell}$ 
is equivalent to the statement that the $\ell$-th unfolding $\mbf{X}^{\{\ell\}}$ is a rank-one matrix, i.e., that all its order-two minors vanish, for all $\ell \in \uglate{3}$. 

In order to define relaxations of the unit tensor nuclear norm ball we introduce the polynomial ideal ${J}_3 \subset\R\uglate{\mbf{x}}=\R\uglate{x_{111},x_{112,}\ldots, x_{n_1 n_2 n_3}}$ as the one generated by
\begin{equation}\label{TH:groebner3}
\mathcal{G}_{3}= \skup{f^{\zag{\bm{\alpha},\bm{\beta}}}\zag{\mbf{x}} : \zag{\bm{\alpha},\bm{\beta}}  \in \mathcal{S}} \cup \skup{g_3\zag{\mbf
{x}}},
\end{equation}
i.e., ${J}_3=\interval{\mathcal{G}_{3}}$. Its real algebraic variety equals the set of rank-one third-order tensors with unit Frobenius norm and its convex hull coincides
with the unit tensor nuclear norm ball. The next result provides the Gr{\"o}bner
basis of ${J}_3$.

\begin{theorem}\label{th3rd}
The basis $\mathcal{G}_{3}$ defined in \eqref{TH:groebner3} forms the reduced Gr{\"o}bner basis of the ideal ${J}_3=\interval{\mathcal{G}_{3}}$ with respect to the grevlex order.
\end{theorem}

\begin{proof} Similarly to the proof of Theorem~\ref{Mgroebner} we need to show that $S\zag{p,q} \rightarrow_{\mathcal{G}_{3}} 0$ for all
polynomials $p,q \in \mathcal{G}_{3}$ whose leading terms are not relatively prime.
The leading monomials with respect to the grevlex ordering are given by 
\begin{align*}
& \LM(g_3) =  x_{111}^2 \\
\text{and } & \LM(f^{\zag{\bm{\alpha},\bm{\beta}}})=x_{\bm{\alpha}}x_{\bm{\beta}}, \quad \zag{\bm{\alpha},\bm{\beta}} \in \mathcal{S}.
\end{align*}
The leading terms of $g_3$ and $f^{\zag{\bm{\alpha},\bm{\beta}}}$ are always relatively prime. 
 First  we consider two distinct polynomials $f,g \in \{f^{\zag{\bm{\alpha},\bm{\beta}}}: \zag{\bm{\alpha},\bm{\beta}}  \in \mathcal{S}_3\}$. 
Let $f=f^{\zag{\bm{\alpha},\bm{\beta}}}$ and $g=f^{\zag{\bm{\alpha},\overline{\bm{\beta}}}}$ for $\zag{\bm{\alpha},\bm{\beta}}  \in \overline{\mathcal{S}}_3$, where $\overline{\bm{\beta}}=\zag{\beta_1,\alpha_2,\beta_3}$. That is,
\begin{align*}
f(\mbf{x})=x_{\bm{\alpha}}x_{\bm{\beta}} - x_{\bm{\alpha} \vee \bm{\beta}}x_{\bm{\alpha} \wedge \bm{\beta}}, \quad\quad
g(\mbf{x})=x_{\bm{\alpha}}x_{\overline{\bm{\beta}}} - x_{\bm{\alpha} \vee \overline{\bm{\beta}}}x_{\bm{\alpha} \wedge \overline{\bm{\beta}}}. 
\end{align*}
Since $\bm{\alpha} \wedge \bm{\beta}=\bm{\alpha} \wedge \overline{\bm{\beta}}$ and $f^{\zag{\bm{\beta},\bm{\alpha} \vee \overline{\bm{\beta}}}} \in \{f^{\zag{\bm{\alpha},\bm{\beta}}}: \zag{\bm{\alpha},\bm{\beta}} \in \mathcal{S}_2\}$, then 
$$S\zag{f,g}=x_{\bm{\alpha} \wedge \bm{\beta}}\zag{-x_{\overline{\bm{\beta}}}x_{\bm{\alpha} \vee \bm{\beta}} + x_{\bm{\beta}}x_{\bm{\alpha} \vee \overline{\bm{\beta}}}} = x_{\bm{\alpha} \wedge \bm{\beta}} f^{\zag{\bm{\beta},\bm{\alpha} \vee \overline{\bm{\beta}}}} \rightarrow_{\mathcal{G}_{3}} 0.$$

Next we show that $S\zag{f,g}\in {J}_3$, for $f \in \skup{f^{\zag{\bm{\alpha},\bm{\beta}}}: \zag{\bm{\alpha},\bm{\beta}} \in \mathcal{S}_2}$ and $g \in \skup{f^{\zag{\bm{\alpha},\bm{\beta}}}: \zag{\bm{\alpha},\bm{\beta}} \in \mathcal{S}_1}$. Let $f=f^{\zag{\bm{\alpha},\hat{\bm{\beta}}}}$ with $\hat{\bm{\beta}}=\zag{\alpha_1,\beta_2,\beta_3}$ and $g = f^{\zag{\bm{\alpha},\tilde{\bm{\beta}}}}$ with $\tilde{\bm{\beta}}=\zag{\beta_1,\beta_2,\alpha_3}$, where $\zag{\bm{\alpha},\bm{\beta}} \in \overline{\mathcal{S}}_2$. 
Since $x_{\bm{\alpha} \wedge \hat{\bm{\beta}}}=x_{\bm{\alpha} \wedge \tilde{\bm{\beta}}}$, $f^{\zag{\hat{\bm{\beta}}, \bm{\alpha} \vee \tilde{\bm{\beta}}}} \in \skup{f^{\zag{\bm{\alpha},\bm{\beta}}}:\zag{\bm{\alpha},\bm{\beta}} \in \mathcal{S}_3}$, and $f^{\zag{\bm{\alpha} \vee \hat{\bm{\beta}}, \tilde{\bm{\beta}} }} \in \skup{f^{\zag{\bm{\alpha},\bm{\beta}}}:\zag{\bm{\alpha},\bm{\beta}} \in \mathcal{S}_1}$
$$S\zag{f,g}=x_{\bm{\alpha} \wedge \hat{\bm{\beta}}}\zag{-x_{\tilde{\bm{\beta}}}x_{\bm{\alpha} \vee \hat{\bm{\beta}}}+x_{\hat{\bm{\beta}}}x_{\bm{\alpha} \vee \tilde{\bm{\beta}}}}=x_{\bm{\alpha} \wedge \hat{\bm{\beta}}}\zag{ f^{\zag{\hat{\bm{\beta}}, \bm{\alpha} \vee \tilde{\bm{\beta}} }} - f^{\zag{\bm{\alpha} \vee \hat{\bm{\beta}}, \tilde{\bm{\beta}} }}} \rightarrow_{\mathcal{G}_{3}} 0.$$
For the remaining cases one proceeds similarly.
In order to show that $\mathcal{G}_{3}$ is the reduced Gr{\"o}bner basis, one uses the same arguments as in the proof of Theorem~\ref{Mgroebner}.
\end{proof}

\begin{remark}
The above Gr{\"o}bner basis $\mathcal{G}_{3}$ is obtained by taking a particular subset of all order-two minors of all three unfoldings of the tensor $\mbf{X} \in \R^{n_1 \times n_2 \times n_3}$ (not considering the same minor twice). One might think that the $\theta_1$-norm obtained in this way corresponds 
to a (weighted) sum of the nuclear norms of the unfoldings, which has been used in \cite{gareya11,gohumuwr14} for tensor recovery. 
The examples of cubic tensors $\mbf{X} \in \R^{2 \times 2 \times 2}$ presented in Table \ref{TableEx} show that this is not the case.
\begin{table}
\caption{Matrix nuclear norms of unfoldings and $\theta_1$-norm of tensors $\mbf{X} \in \R^{2 \times 2 \times 2}$, which are represented in the second column  as $\mbf{X}=\left[\mbf{X}\zag{:,:,1} \, \middle\vert \, \mbf{X}\zag{:,:,2}\right]$. The third, fourth and fifth column represent the nuclear norms of the first, second and the third unfolding of a tensor $\mbf{X}$, respectively. The last column contains the numerically computed $\theta_1$-norm.}
\label{TableEx}
\begin{tabular}{ l l c c c c }
\hline\noalign{\smallskip}
 & $\mbf{X} \in \R^{2 \times 2 \times 2}$ & $\|\mbf{X}^{\{1\}}\|_*$ & $\|\mbf{X}^{\{2\}}\|_*$ &  $\|\mbf{X}^{\{3\}}\|_*$ &  $\|\mbf{X}\|_{\theta_1}$ \\
\noalign{\smallskip}\hline\noalign{\smallskip}
\rule{0pt}{4ex} $1$ & $\left[\begin{matrix} 1& 0 \\ 0 & 0 \end{matrix} \, \middle\vert \, \begin{matrix} 0 & 0 \\ 0 & 1\end{matrix}\right]$ &   $2$ & $2$ & $2$ & $2$ \\ 
\rule{0pt}{4ex}  $2$&  $\left[\begin{matrix} 1& 0 \\ 0 & 1 \end{matrix} \, \middle\vert \, \begin{matrix} 0 & 0 \\ 0 & 0\end{matrix}\right]$  & $2$ & $2$ & $\sqrt{2}$ & $2$\\
\rule{0pt}{4ex}    $3$& $\left[\begin{matrix} 1& 0 \\ 0 & 0 \end{matrix} \, \middle\vert \, \begin{matrix} 0 & 0 \\ 1 & 0\end{matrix}\right]$ &  $2$ & $\sqrt{2}$ & $2$ & $2$\\
\rule{0pt}{4ex}    $4$& $\left[\begin{matrix} 1& 0 \\ 0 & 0 \end{matrix} \, \middle\vert \, \begin{matrix} 0 & 1 \\ 0 & 0\end{matrix}\right]$ &
 $\sqrt{2}$ & $2$ & $2$ & $2$\\
\rule{0pt}{4ex}    $5$& $\left[\begin{matrix} 1& 0 \\ 0 & 1 \end{matrix} \, \middle\vert \, \begin{matrix} 0 & 1 \\ 0 & 0\end{matrix}\right]$ &
 $\sqrt{2}+1$ & $\sqrt{2}+1$ & $\sqrt{2}+1$ & $3$\\
 \noalign{\smallskip}\hline
\end{tabular}
\end{table}
Assuming that $\theta_1$-norm is a linear combination of the nuclear norm of the unfoldings, there exist $\alpha$, $\beta$, $\gamma \in \R$ such that
$ \alpha \|\mbf{X}^{\{1\}}\|_*+ \beta \|\mbf{X}^{\{2\}}\|_* +\gamma \|\mbf{X}^{\{3\}}\|_*=\|\mbf{X}\|_{\theta_1}.$ From the first and the second tensor in Table \ref{TableEx} we obtain $\gamma=0$. Similarly, the first and the third tensor, and the first and the fourth tensor give $\beta=0$  and $\alpha=0$, respectively. Thus, the $\theta_1$-norm does not coincide with 
a weighted sum of the nuclear norms of the unfoldings. In addition, the last tensor shows that the $\theta_1$-norm does not equal maximum of the norms of the unfoldings.
\end{remark}

Theorem~\ref{th3rd} states that $\mathcal{G}_3$ is the reduced Gr{\"o}bner basis of the ideal $J_3$ generated by all order-two minors of all matricizations of an order-three tensor. That is, $J_3$ is generated by the following polynomials
\begin{align*}
f_{\zag{\bm{\alpha},\bm{\beta}}}^{\{1\}}(\mbf{x})&=-x_{\alpha_1 \alpha_2 \alpha_3} x_{\beta_1 \beta_2 \beta_3} + x_{\alpha_1 \beta_2 \beta_3} x_{\beta_1 \alpha_2 \alpha_3}, \quad \text{for } (\bm{\alpha},\bm{\beta}) \in \bm{\mathcal{T}}^{\{1\}} \\
f_{\zag{\bm{\alpha},\bm{\beta}}}^{\{2\}}(\mbf{x})&=-x_{\alpha_1 \alpha_2 \alpha_3} x_{\beta_1 \beta_2 \beta_3} + x_{\beta_1 \alpha_2 \beta_3} x_{\alpha_1 \beta_2 \alpha_3}, \quad \text{for } (\bm{\alpha},\bm{\beta}) \in \bm{\mathcal{T}}^{\{2\}}\\
f_{\zag{\bm{\alpha},\bm{\beta}}}^{\{3\}}(\mbf{x})&=-x_{\alpha_1 \alpha_2 \alpha_3} x_{\beta_1 \beta_2 \beta_3} + x_{\beta_1 \beta_2 \alpha_3} x_{\alpha_1 \alpha_2 \beta_3}, \quad \text{for } (\bm{\alpha},\bm{\beta}) \in \bm{\mathcal{T}}^{\{3\}},
\end{align*}
where $\skup{f_{\zag{\bm{\alpha},\bm{\beta}}}^{\{k\}}(\mbf{x}):\zag{\bm{\alpha},\bm{\beta}} \in \bm{\mathcal{T}}^{\{k\}}}$ is the set of all order-two minors of the $k$th unfolding and 
\begin{align*}
\bm{\mathcal{T}}^{\{k\}}&=\skup{(\bm{\alpha},\bm{\beta}): \alpha_k\neq\beta_k, \, \overline{\bm{\alpha}} \neq\overline{\bm{\beta}}, \text{where } \overline{\alpha}_k=\overline{\beta}_k=0, \overline{\alpha}_{\ell}=\alpha_{\ell}, \overline{\beta}_{\ell}=\beta_{\ell}}.
\end{align*}
 For  $\zag{\bm{\alpha},\bm{\beta}}$,  $x_{\bm{\alpha}^{\{k\}}} x_{\bm{\beta}^{\{k\}}}$ denotes a monomial where $\alpha_k^{\{k\}}=\alpha_k$, $\beta_k^{\{k\}}=\beta_k$, and $\alpha_{\ell}^{\{k\}}=\beta_{\ell}$, $\beta_{\ell}^{\{k\}}=\alpha_{\ell}$, for all $\ell \in \uglate{d}\backslash \{k\}$.  
Notice that $f_{\zag{\bm{\alpha},\bm{\beta}}}^{\{k\}}(\mbf{x})=f_{\zag{\bm{\beta},\bm{\alpha}}}^{\{k\}}(\mbf{x})=-f_{\zag{\bm{\alpha}^{\{k\}},\bm{\beta}^{\{k\}}}}^{\{k\}}(\mbf{x})=-f_{\zag{\bm{\beta}^{\{k\}},\bm{\alpha}^{\{k\}}}}^{\{k\}}(\mbf{x})$, for all $\zag{\bm{\alpha},\bm{\beta}} \in \bm{\mathcal{T}}^{\{k\}}$, and all $k \in \uglate{3}$.  
Let us now consider a TT-format and a corresponding notion of tensor rank. Recall that a TT-rank of an order three tensor is a vector $\mbf{r}=(r_1,r_2)$ where $r_1=\rank(\mbf{X}^{\{1\}})$ and $r_2=\rank(\mbf{X}^{\{1,2\}})$. Consequently, we consider an ideal $J_{3,\text{TT}}$ generated by all order-two minors of matricizations $\mbf{X}^{\{1\}}$ and $\mbf{X}^{\{1,2\}}$ of the order-$3$ tensor. That is, the ideal $J_{3,\text{TT}}$ is generated by the polynomials
\begin{align*}
f_{\zag{\bm{\alpha},\bm{\beta}}}^{\{1\}}(\mbf{x})&=-x_{\alpha_1 \alpha_2 \alpha_3} x_{\beta_1 \beta_2 \beta_3} + x_{\alpha_1 \beta_2 \beta_3} x_{\beta_1 \alpha_2 \alpha_3}, \quad \text{for } (\bm{\alpha},\bm{\beta}) \in \bm{\mathcal{T}}^{\{1\}}, \\
f_{\zag{\bm{\alpha},\bm{\beta}}}^{\{1,2\}}(\mbf{x})&=-x_{\alpha_1 \alpha_2 \alpha_3} x_{\beta_1 \beta_2 \beta_3} + x_{\alpha_1 \alpha_2 \beta_3} x_{\beta_1 \beta_2 \alpha_3} , \quad \text{for } (\bm{\alpha},\bm{\beta}) \in \bm{\mathcal{T}}^{\{1,2\}},
\end{align*}
where $\bm{\mathcal{T}}^{\{1,2\}}=\skup{(\bm{\alpha},\bm{\beta}): \zag{\alpha_1,\alpha_2,0} \neq \zag{\beta_1,\beta_2,0}, \, \alpha_3 \neq\beta_3}$. 
\begin{theorem}
The polynomial ideals $J_3$ and $J_{3,\text{TT}}$ are equal. 
\end{theorem}
\begin{remark}
As a consequence, $\mathcal{G}_3$ is also the reduced Gr{\"o}bner basis for  the ideal $J_{3,\text{TT}}$ with respect to the grevlex ordering.
\end{remark}
\begin{proof}
Notice that $\zag{\mbf{X}^{\{3\}}}^T=\mbf{X}^{\{1,2\}}$ and therefore $$\skup{f_{\zag{\bm{\alpha},\bm{\beta}}}^{\{3\}}(\mbf{x}):\zag{\bm{\alpha},\bm{\beta}} \in \bm{\mathcal{T}}^{\{3\}}} = \skup{f_{\zag{\bm{\alpha},\bm{\beta}}}^{\{1,2\}}(\mbf{x}):\zag{\bm{\alpha},\bm{\beta}} \in \bm{\mathcal{T}}^{\{1,2\}}}.$$
Hence, it is enough to show that $f_{\zag{\bm{\alpha},\bm{\beta}}}^{\{2\}} \in J_{3,\text{TT}}$, for all $\zag{\bm{\alpha},\bm{\beta}} \in \bm{\mathcal{T}}^{\{2\}}$. By definition of $\bm{\mathcal{T}}^{\{2\}}$, we have that
$\alpha_2\neq\beta_2$ and $(\alpha_1,0,\alpha_3) \neq (\beta_1,0,\beta_3)$. We can assume that $\alpha_3\neq \beta_3$, since otherwise $f_{\zag{\bm{\alpha},\bm{\beta}}}^{\{2\}} = f_{\zag{\bm{\alpha},\bm{\beta}}}^{\{1\}}$. Analogously,  $\alpha_1 \neq \beta_1$ since otherwise $f_{\zag{\bm{\alpha},\bm{\beta}}}^{\{2\}} = f_{\zag{\bm{\alpha},\bm{\beta}}}^{\{1,2\}}$.
  Consider the following polynomials 
\begin{align*}
f(\mbf{x})&=-x_{\alpha_1 \alpha_2 \alpha_3} x_{\beta_1 \beta_2 \beta_3} + x_{\beta_1 \alpha_2 \beta_3} x_{\alpha_1 \beta_2 \alpha_3}, \quad  (\bm{\alpha},\bm{\beta}) \in\bm{\mathcal{T}}^{\{2\}}\\
g(\mbf{x})&=-x_{\beta_1 \beta_2 \alpha_3} x_{\alpha_1 \alpha_2 \beta_3} + x_{\beta_1 \alpha_2 \beta_3} x_{\alpha_1 \beta_2 \alpha_3} ,\quad  (\beta_1,\beta_2,\alpha_3,\alpha_1,\alpha_2,\beta_3) \in\bm{\mathcal{T}}^{\{1\}}\\
h(\mbf{x})&=-x_{\alpha_1 \alpha_2 \alpha_3} x_{\beta_1 \beta_2 \beta_3} + x_{\alpha_1 \alpha_2 \beta_3} x_{\beta_1 \beta_2 \alpha_3} ,\quad  (\bm{\alpha},\bm{\beta}) \in\bm{\mathcal{T}}^{\{1,2\}}.
\end{align*}
Thus, we have that $f(\mbf{x})=g(\mbf{x})+h(\mbf{x}) \in J_{3,\text{TT}}$.
\end{proof}

\subsection{The theta norm for general $d$th-order tensors}\label{dtensor}

Let us now consider $d$th-order tensors in $\R^{n_1\times n_2 \times \cdots \times n_d}$ for general $d \geq 4$.
Our approach relies again on the fact that a tensor $\mbf{X} \in \Rd$ is of rank-one if and only if all its matricizations are rank-one matrices, or equivalently, 
if all minors of order two of each matricization vanish. 

The description of the polynomial ideal generated by the second order minors of all matricizations of a tensor $\mbf{X} \in \Rd$ 
unfortunately requires some technical notation. Again, we do not need all such minors in the generating set that we introduce next.
In fact, this generating set will turn out to be the reduced Gr{\"o}bner basis of the ideal.

Similarly to before, the entry $\zag{\alpha_1,\alpha_2,\ldots,\alpha_d}$ of a tensor $\mbf{X} \in \Rd$ corresponds to the variable $x_{\alpha_1\alpha_2 \cdots \alpha_d}$ or simply $x_{\bm{\alpha}}$. 
We aim at introducing a set of polynomials of the form
\begin{equation}\label{TH:defpol}
 f_{d}^{\zag{\bm{\alpha},\bm{\beta}}}(\mbf{x}):= -x_{\bm{\alpha} \wedge \bm{\beta}}x_{\bm{\alpha} \vee \bm{\beta}}+x_{\bm{\alpha}}x_{\bm{\beta}}
\end{equation}
which will generate the desired polynomial ideal. These polynomials 
correspond to a subset of all order-two minors of all the possible $d$th-order tensor matricizations.  
The set $\mathcal{S}$ denotes the indices where $\bm{\alpha}$ and $\bm{\beta}$ differ. Since for an order-two minor of a matricization $\mbf{X}^{\mathcal{M}}$ the sets $\bm{\alpha}$ and $\bm{\beta}$ need to differ in at least two indices, $\mathcal{S}$ is contained in  
$$ \mathcal{S}_{\uglate{d}}:=\{\mathcal{S} \subset \uglate{d}: 2\leq | \mathcal{S} | \leq d\}.$$ 
Given the set $\mathcal{S}$ of different indices, we require all non-empty subsets $\mathcal{M} \subset \mathcal{S}$ of possible indices which are ``switched'' between
$\bm{\alpha}$ and $\bm{\beta}$ for forming the minors in \eqref{TH:defpol}. This implies that, without loss of generality, 
\begin{align*}
&\alpha_j > \beta_j, \quad \text{ for all } j \in \mathcal{M}\\
&\alpha_k < \beta_k, \quad \text{ for all } k \in \mathcal{S} \backslash \mathcal{M}.
\end{align*}
That is, the same minor is obtained if we require that $\alpha_j<\beta_j$ for all $j \in \mathcal{M}$ and $\alpha_k>\beta_k$ for all $k \in \mathcal{S}\backslash \mathcal{M}$ since the set of all two-minors of $\mbf{X}^{\mathcal{M}}$ coincides with the set of all two-minors of $\mbf{X}^{\mathcal{S}\backslash \mathcal{M}}$.

For $\mathcal{S} \in \mathcal{S}_{\uglate{d}}$, we define $e_{\mathcal{S}}:=\min\{p: p \in \mathcal{S}\}$.  The set $\mathcal{M}$ corresponds to an associated
matricization $\mbf{X}^{\mathcal{M}}$. 
The set of possible subsets $\mathcal{M}$ is given as 
\begin{equation*}
\mathcal{P}_{\mathcal{S}}=
\begin{cases}
\skup{\mathcal{M} \subset \mathcal{S}: |\mathcal{M}| \leq \floor{\frac{|\mathcal{S}|}{2}}}\backslash\{\eset\}, & \hspace{-0.7cm}\text{if } |\mathcal{S}| \text{ is odd}, \\
\skup{\mathcal{M} \subset \mathcal{S}: |\mathcal{M}| \leq \floor{\frac{|\mathcal{S}|-1}{2}}} \cup \skup{ \mathcal{M} \subset \mathcal{S}: | \mathcal{M}|=\frac{|\mathcal{S}|}{2}, e_{\mathcal{S}} \in \mathcal{M}}\backslash\{\eset\}, &\hspace{-0.3cm} \text{otherwise}.
\end{cases}
\end{equation*} 
Notice that $\mathcal{P}_{\mathcal{S}} \cup \mathcal{P}_{\mathcal{S}^c} \cup \{\eset\} \cup \mathcal{S}$ 
with $\mathcal{P}_{\mathcal{S}^c}:=\{\mathcal{M}: \mathcal{S} \backslash \mathcal{M} \in \mathcal{P}_{\mathcal{S}}\}$ forms the power set 
of $\mathcal{S}$. The constraint on the size of $\mathcal{M}$ in the definition of $\mathcal{P}_{\mathcal{S}}$ is motivated by
the fact that the role of $\bm{\alpha}$ and $\bm{\beta}$ can be switched and lead to the same
polynomial $f_{d}^{(\bm{\alpha},\bm{\beta})}$.

Thus, for $\mathcal{S} \in \mathcal{S}_{\uglate{d}}$ and $\mathcal{M} \in \mathcal{P}_{\mathcal{S}}$, we define a set
\begin{align*}
 \mathcal{T}_{d}^{\mathcal{S},\mathcal{M}}:=\{\zag{\bm{\alpha},\bm{\beta}}: \, & \alpha_i=\beta_i, \text{ for all } i \notin \mathcal{S}   \\
& \alpha_j > \beta_j , \text{ for all } j \in \mathcal{M}   \\
& \alpha_k < \beta_k , \text{ for all } k \in \mathcal{S}\backslash\mathcal{M}  \}    .
\end{align*}
For  notational purposes, we define 
$$\{f_d^{\mathcal{S}}\}=\cup_{\mathcal{M} \in \mathcal{P}_{\mathcal{S}}}\{f_{d}^{\zag{\bm{\alpha},\bm{\beta}}}: \zag{\bm{\alpha},\bm{\beta}} \in \mathcal{T}_d^{\mathcal{S},\mathcal{M}}\} \quad \text{ for } \mathcal{S} \in \mathcal{S}_{\uglate{d}}.
$$
Since we are interested in unit Frobenius norm tensors, we also introduce the polynomial
$$ 
g_d\zag{\mbf{x}}=\sum_{i_1=1}^{n_1}\sum_{i_2=1}^{n_2}\ldots \sum_{i_d=1}^{n_d}x_{i_1 i_2 \ldots i_d}^2-1.
$$
Our polynomial ideal is then the one generated by the polynomials in
\[
\mathcal{G}_d = \bigcup_{\mathcal{S}\in \mathcal{S}_{\uglate{d}}} \{f_d^{\mathcal{S}}\} \cup \{g_d\} \subset \R\uglate{\mbf{x}}=\R\uglate{x_{11\ldots1}, x_{11 \ldots 2}, \ldots, x_{n_1 n_2 \ldots n_d}},
\] i.e.,
$J_{d} = \langle \mathcal{G}_d \rangle$.
As in the special case of the third-order tensors, not all second order minors corresponding to 
all matricizations are contained in the generating set $\mathcal{G}_d$ due to the condition $i_k < \hat{i}_k$ for all $k \in \mathcal{S}$
in the definition of $\mathcal{T}_d^{\mathcal{S}}$.
Nevertheless all second order minors 
are contained in the ideal $J_{d}$ as will also be revealed by the proof of Theorem~\ref{thmmadth} below. 
For instance, $h(\mbf{x})=-x_{1234}x_{2343}+x_{1243}x_{2334}$ -- corresponding to a minor of the matricization $\mathbf{X}^{\mathcal{M}}$ for
$\mathcal{M} = \{1,2\}$ -- does not belong to $\mathcal{G}_4$, but it does belong to the ideal $J_{4}$.
Moreover, it is straightforward to verify that all polynomials in $\mathcal{G}_d$ differ from each other.

The algebraic variety of $J_{d}$ consists of all rank-one unit Frobenius norm order-$d$ tensors as desired,  and its convex
hull yields the tensor nuclear norm ball.

\begin{theorem}\label{thmmadth}
The set $\mathcal{G}_d$ forms 
the reduced Gr{\"o}bner basis of the ideal $J_{d}$ with respect to the grevlex order.
\end{theorem}

\begin{proof}
Again, we use Buchberger's criterion stated in Theorem~\ref{GrS}.
First notice that the polynomials $g_d$ and $f_{d}^{\zag{\bm{\alpha},\bm{\beta}}}$ are always relatively prime, 
since $\LM(g_d)=x_{11\ldots 1}^2$ and $\LM(f_{d}^{\zag{\bm{\alpha},\bm{\beta}}})=x_{\bm{\alpha}} x_{\bm{\beta}}$ for  $(\bm{\alpha},\bm{\beta}) \in \mathcal{T}_d^{\mathcal{M},\mathcal{S}}$, where $\mathcal{S} \in \mathcal{S}_{\uglate{d}}$ and $\mathcal{M} \in \mathcal{P}_{\mathcal{S}}$. Therefore, we need to show that $S(f_1,f_2) \rightarrow_{\mathcal{G}_d} 0$, for all $f_1,f_2 \in \mathcal{G}_d\backslash \{g_d\}$ with 
$f_1\neq f_2$. To this end, we analyze the division algorithm on $\left\langle\mathcal{G}_d\right\rangle$. 

Let $f_1,f_2 \in \mathcal{G}_d$ with $f_1 \neq f_2$. Then it holds $\LM(f_1) \neq \LM(f_2)$. 
If these leading monomials are not relatively prime, the $S$-polynomial is of the form $$S(f_1,f_2)=x_{\bm{\alpha}^1} x_{\bm{\alpha}^2} x_{\bm{\alpha}^3} - x_{\bar{\bm{\alpha}}^1} x_{\bar{\bm{\alpha}}^2}x_{\bar{\bm{\alpha}}^3} $$ with $\skup{\alpha_k^1, \alpha_k^2, \alpha_k^3}= \skup{\bar{\alpha}_k^1,\bar{\alpha}_k^2, \bar{\alpha}_k^3}$ for all $k \in \uglate{d}$.

The step-by-step procedure of the division algorithm for our scenario is presented in Algorithm~\ref{Division}. We will show that the algorithm eventually stops
and that step 2) is feasible, i.e., that there always exist $k$ and $\ell$ such that line 7 of Algorithm~\ref{Division}  holds -- provided that $S^i \neq 0$.
(In fact, the purpose of the algorithm is to achieve the condition that in the $i$th iteration of the algorithm $\hat{\alpha}_k^{1,i} \leq \hat{\alpha}_k^{2,i} \leq \hat{\alpha}_k^{3,i}$,
 for all $k \in \uglate{d}$.)  
 This will show then that 
$S(f_1,f_2) \rightarrow_{\mathcal{G}_d} 0$.

\begin{algorithm}[h!]
\caption{The division algorithm on the ideal $\interval{\mathcal{G}_d}$.} 
\label{Division}
\begin{algorithmic}
\REQUIRE polynomials $f_1,f_2 \in \mathcal{G}_d$
\STATE $S^0 = S(f_1,f_2)=x_{\bm{\alpha}^1}x_{\bm{\alpha}^2}x_{\bm{\alpha}^3}-x_{\bar{\bm{\alpha}}^1} x_{\bar{\bm{\alpha}}^2} x_{\bar{\bm{\alpha}}^3}$, $i = 0$
\WHILE{$S^i \neq 0$}
\STATE $1)$ \bf{Let} $\LM(S^i)=x_{\hat{\bm{\alpha}}^{1,i}} x_{\hat{\bm{\alpha}}^{2,i}} x_{\hat{\bm{\alpha}}^{3,i}}$ \textbf{and} $\NLM(S^i)=\aps{S^i-\LT(S^i)}$  
\STATE $2)$ \bf{Find indices} $\bm{\alpha}^{1,i}, \bm{\alpha}^{2,i} \in \{\hat{\bm{\alpha}}^{1,i},\hat{\bm{\alpha}}^{2,i},\hat{\bm{\alpha}}^{3,i}\}$ \textbf{such that there exist}  at least one $k$ and at least one $\ell$ for which 
$${\alpha}_k^{1,i} < {\alpha}_k^{2,i} \quad \text{ and } \quad {\alpha}_{\ell}^{1,i} > {\alpha}_{\ell}^{2,i} \quad \text{s.t.\ } \quad \mathcal{M}_i:=\skup{\ell \in \uglate{d}: \alpha_{\ell}^{1,i} > \alpha_{\ell}^{2,i}} \in \mathcal{P}_{\mathcal{S}},$$
 where  $\mathcal{S}:=\skup{k \in \uglate{d}: \alpha_k^{1,i} \neq \alpha_k^{2,i}}$ 
 and let $\bm{\alpha}^{3,i}$ be the remaining index in $\{\hat{\bm{\alpha}}^{1,i},\hat{\bm{\alpha}}^{2,i},\hat{\bm{\alpha}}^{3,i}\} \backslash \{\bm{\alpha}^{1,i},\bm{\alpha}^{2,i}\}$.  
\STATE$3)$ Divide $S^i$ by $f_{d}^{\zag{\bm{\alpha}^{1,i}, \bm{\alpha}^{2,i}}}=x_{\bm{\alpha}^{1,i}}x_{\bm{\alpha}^{2,i}}-x_{\bm{\alpha}^{1,i} \wedge \bm{\alpha}^{2,i}}x_{\bm{\alpha}^{1,i} \vee \bm{\alpha}^{2,i}}$ to obtain 
 \begin{align*}
S^{i}&=\LC(S^i)\big[x_{\bm{\alpha}^{3,i}}(-x_{\bm{\alpha}^{1,i} \wedge \bm{\alpha}^{2,i}}x_{\bm{\alpha}^{1,i} \vee \bm{\alpha}^{2,i}}+x_{\bm{\alpha}^{1,i}}x_{\bm{\alpha}^{2,i}}) \\
 & \hspace{4cm}+x_{\bm{\alpha}^{1,i} \wedge \bm{\alpha}^{2,i}}x_{\bm{\alpha}^{1,i} \vee \bm{\alpha}^{2,i}}x_{\bm{\alpha}^{3,i}}-\NLM(S^i)\big].
 \end{align*}
\STATE $4)$ Define 
$$S^{i+1}:=x_{\bm{\alpha}^{1,i} \wedge \bm{\alpha}^{2,i}}x_{\bm{\alpha}^{1,i} \vee \bm{\alpha}^{2,i}}x_{\bm{\alpha}^{3,i}}-\NLM(S^i).$$ 
\STATE$5)$ $i=i+1$ 
\ENDWHILE
\end{algorithmic}
\end{algorithm}

Before passing to the general proof, we illustrate the division algorithm on an example for $d=4$. The experienced reader
may skip this example.

\medskip

Let  
$f_1(\mbf{x}):=f_{4}^{(1212,2123)}(\mbf{x})=-x_{1112}x_{2223}+x_{1212}x_{2123} \in \mathcal{G}_4$ (with the corresponding sets $\mathcal{S}=\{1,2,3,4\}$, $\mathcal{M}=\{2\}$) and  
$f_2(\mbf{x}):=f_{4}^{(3311,2123)}(\mbf{x})=-x_{2111}x_{3323}+x_{3311}x_{2123} \in \mathcal{G}_4$ (with the corresponding sets  $\mathcal{S}=\{1,2,3,4\}$, $\mathcal{M}=\{1,2\}$). 
We will show that $S(f_1,f_2)=-x_{1112}x_{2223}x_{3311}+x_{1212}x_{2111}x_{3323} \rightarrow_{\mathcal{G}_4} 0$ by going through the division algorithm.

In iteration $i=0$ we set $S^0=S(f_1,f_2)=-x_{1112}x_{2223}x_{3311}+x_{1212}x_{2111}x_{3323}$. The leading monomial is 
$\LM(S^0)=x_{1112}x_{2223}x_{3311}$, the leading coefficient is $\LC(S^0)=-1$, and the non-leading monomial is $\NLM(S^0)=x_{1212}x_{2111}x_{3323}$. 
Among the two options for choosing a pair of indexes $(\bm{\alpha}^{1,0},\bm{\alpha}^{2,0})$ in step 2), we decide to take
$\bm{\alpha}^{1,0}=1112$ and $\bm{\alpha}^{2,0}=3311$ which leads to the set 
$\mathcal{M}_0 =\{4\}$. 
The polynomial $x_{\bm{\alpha}^{1,0}}x_{\bm{\alpha}^{2,0}}-x_{\bm{\alpha}^{1,0} \wedge \bm{\alpha}^{2,0}}x_{\bm{\alpha}^{1,0} \vee \bm{\alpha}^{2,0}}$ then equals the polynomial $ f_{4}^{(1112,3311)}(\mbf{x})=-x_{1111}x_{3312}+x_{1112}x_{3311} \in \mathcal{G}_4$ 
and we can write
$$S^0=-1\cdot \Big(x_{2223}\zag{-x_{1111}x_{3312}+x_{1112}x_{3311}}+\underbrace{x_{1111}x_{2223}x_{3312}-x_{1212}x_{2111}x_{3323}}_{\displaystyle =S^1}\Big) .$$
The leading and non-leading monomials of $S^1$ are $\LM(S^1)=x_{1111}x_{2223}x_{3312}$ and $\NLM(S^1)=x_{1212}x_{2111}x_{3323}$, respectively,
while $\LC(S^1)=1$. The only option for a pair of indices as in line 7 of Algorithm~\ref{Division}   
is $\bm{\alpha}^{1,1}=3312,\bm{\alpha}^{2,1}=2223$, so that the set 
$\mathcal{M}_1=\{1,2\}$. 
The divisor $x_{\bm{\alpha}^{1,1}}x_{\bm{\alpha}^{2,1}}-x_{\bm{\alpha}^{1,1} \wedge \bm{\alpha}^{2,1}}x_{\bm{\alpha}^{1,1} \vee \bm{\alpha}^{2,1}}$ in the step $4)$ equals $f_{4}^{(3312,2223)}(\mbf{x})= -x_{2212}x_{3323}+x_{3312}x_{2223}\in \mathcal{G}_4$  and we obtain 
$$
S^1=1\cdot \Big(x_{1111}\zag{-x_{2212}x_{3323}+x_{2223}x_{3312}}+ \underbrace{x_{1111}x_{2212}x_{3323}-x_{1212}x_{2111}x_{3323}}_{\displaystyle = S^2}\Big) . 
$$
The index sets of the monomial $x_{\bm{\alpha}^1}x_{\bm{\alpha}^2}x_{\bm{\alpha}^3}=x_{1111}x_{2212}x_{3323}$ in $S^2$ satisfy
$$\alpha_k^1 \leq \alpha_k^2 \leq \alpha_k^3 \quad \text{ for all } k \in \uglate{4} $$
and therefore it is the non-leading monomial of $S^2$, i.e., $\NLM(S^2)=x_{1111}x_{2212}x_{3323}$. Thus, $\LM(S^2)=x_{1212}x_{2111}x_{3323}$ 
and $\LC(S^2(f_1,f_2))=-1$.
Now the only option for a pair of indices as in step $2)$ is $\bm{\alpha}^{1,2} =2111 $, $ \bm{\alpha}^{2,2} = 1212$ with $\mathcal{M}_2=\{1\}$.
This yields
$$ S^2=-1 \cdot \Big( x_{3323}\zag{-x_{1111}x_{2212}+x_{2111}x_{1212}}+
\underbrace{x_{1111}x_{2212}x_{3323} -x_{1111}x_{2212}x_{3323}}_{\displaystyle = S^3 = 0}\Big). $$
Thus, the division algorithm stops and we obtained after three steps 
\begin{align*}
{S}(f_1,f_2)=S^0& =\LC(S^0)x_{2223}f_{4}^{(1112,3311)}(\mbf{x})+\LC(S^0) \LC(S^1)x_{1111} f_{4}^{(3312,2223)}(\mbf{x}) \\
&  +\LC(S^0)\LC(S^1)\LC(S^2) x_{3323}f_{4}^{(2111,1212)}(\mbf{x}).
\end{align*}
Thus, $S(f_1,f_2) \rightarrow_{\mathcal{G}_4} 0$.

\medskip

Let us now return to the general proof. We first show that there always exist indices $\bm{\alpha}^{1,i}, \bm{\alpha}^{2,i}$ satisfying line 7 of  Algorithm~\ref{Division}  
unless $S^i =0$.
We start by setting $\mbf{x}^{\bm{\alpha}_{i}}=x_{\hat{\bm{\alpha}}^{1,i}}x_{\hat{\bm{\alpha}}^{2,i}}x_{\hat{\bm{\alpha}}^{3,i}}$ 
with $x_{\hat{\bm{\alpha}}^{1,i}} \geq x_{\hat{\bm{\alpha}}^{2,i}} \geq x_{\hat{\bm{\alpha}}^{3,i}}$ to be the leading monomial 
and $\mbf{x}^{\bm{\beta}_{{i}}}$ to be the non-leading  monomial of $S^i$. 
The existence of a polynomial $h \in \mathcal{G}_d$ such that $\LM(h)$ divides 
$\LM(S^i)=x_{\hat{\bm{\alpha}}^{1,i}}x_{\hat{\bm{\alpha}}^{2,i}}x_{\hat{\bm{\alpha}}^{3,i}}=\mbf{x}^{\bm{\alpha}_i}$ is equivalent to the existence of 
$\bm{\alpha}^{1,i},\bm{\alpha}^{2,i} \in \skup{\hat{\bm{\alpha}}^{1,i},\hat{\bm{\alpha}}^{2,i},\hat{\bm{\alpha}}^{3,i}}$ such that there exists at least one $k$ and at least one $\ell$ 
for which ${\alpha}_k^{1,i} < {\alpha}_k^{2,i}$  
and ${\alpha}_{\ell}^{1,i}>{\alpha}_{\ell}^{2,i}$.
If such pair does not exist in iteration $i$, we have 
\begin{equation} \label{TH:index}
\hat{\alpha}_k^{1,i} \leq \hat{\alpha}_k^{2,i} \leq \hat{\alpha}_k^{3,i} \quad \text{ for all } k \in \uglate{d}. 
\end{equation}
We claim that this cannot happen if $S^{i} \neq 0$. In fact, \eqref{TH:index} would imply that the monomial $\mbf{x}^{\bm{\alpha}_{i}}=x_{\hat{\bm{\alpha}}^{1,i}}x_{\hat{\bm{\alpha}}^{2,i}}x_{\hat{\bm{\alpha}}^{3,i}}$ is the smallest monomial $x_{\bm{\beta}} x_{\bm{\gamma}} x_{\bm{\eta}}$ (with respect to the grevlex order) which satisfies 
$$ \skup{\beta_k, \gamma_k, \eta_k}=\{ \hat{\alpha}_k^{1,i},\hat{\alpha}_k^{2,i},\hat{\alpha}_k^{3,i}\} \quad \text{for all } k \in \uglate{d}.$$
However, then $\mbf{x}^{\bm{\alpha}_{i}}$ would not be the leading monomial by definition of the grevlex order, which leads to a contradiction.
Hence, we can always find indices $\bm{\alpha}^{1,i}, \bm{\alpha}^{2,i}$ satisfying line 7 
 in step 2) of  Algorithm~\ref{Division} unless $S^i =0$.

Next we show that the division algorithm always stops in a finite number of steps.
We start with iteration $i=0$ and assume that $S^0 \neq 0$. We choose
$\bm{\alpha}^{1,0},\bm{\alpha}^{2,0},\bm{\alpha}^{3,0}$  as in step 2) of Algorithm~\ref{Division}. Then we divide the polynomial $S^0$  by a 
polynomial $h \in \mathcal{G}_d$ such that $\LM(h)=x_{\bm{\alpha}^{1,0}}x_{\bm{\alpha}^{2,0}}$. The polynomial $h \in \mathcal{G}_d$ is defined 
as in step 3) of the algorithm, i.e.,
$$h(\mbf{x}) =f_{d}^{\zag{\bm{\alpha}^{1,0}, \bm{\alpha}^{2,0}}}=x_{\bm{\alpha}^{1,0}}x_{\bm{\alpha}^{2,0}}- x_{\bm{\alpha}^{1,0} \wedge \bm{\alpha}^{2,0}}x_{\bm{\alpha}^{1,0} \vee \bm{\alpha}^{2,0}}\in \mathcal{G}_d.$$
The division of $S^0$ by $h$ results in
$$ S^0 =\LC(S^0) \Big(x_{{\bm{\alpha}}^{3,0}} \cdot f_{d}^{\zag{\bm{\alpha}^{1,0}, \bm{\alpha}^{2,0}}}
+\underbrace{x_{\bm{\alpha}^{1,0} \wedge \bm{\alpha}^{2,0}}x_{\bm{\alpha}^{1,0} \vee \bm{\alpha}^{2,0}}x_{\bm{\alpha}^{3,0}}- \NLM(S^0)}_{\displaystyle = S^1}\Big).$$
Note that by construction
\begin{equation}\label{TH:cond:I0}
\uglate{\bm{\alpha}^{1,0} \wedge \bm{\alpha}^{2,0}}_k \leq \uglate{\bm{\alpha}^{1,0} \vee \bm{\alpha}^{2,0}}_k \quad \mbox{ for all } k \in \uglate{d}.
\end{equation}

If $S^1 \neq 0$, then in the following iteration $i=1$ we can assume $\LM(S^1)=x_{\bm{\alpha}^{1,0} \wedge \bm{\alpha}^{2,0}} x_{\bm{\alpha}^{1,0} \wedge \bm{\alpha}^{2,0}} x_{\bm{\alpha}^{3,0}}$. 
Due to \eqref{TH:cond:I0}, a pair $\bm{\alpha}^{1,1}, \bm{\alpha}^{2,1}$ as in line 7 of Algorithm~\ref{Division} can be either $\bm{\alpha}^{1,0} \wedge \bm{\alpha}^{2,0}, \bm{\alpha}^{3,0}$ or $\bm{\alpha}^{1,0} \vee \bm{\alpha}^{2,0}, \bm{\alpha}^{3,0}$. 
Let us assume the former.  
Then this iteration results in
$$
S^1=\LC(S^1)\Big(x_{\bm{\alpha}^{3,1}}\cdot f_{d}^{\zag{\bm{\alpha}^{1,1},\bm{\alpha}^{2,1}}} +\underbrace{x_{\bm{\alpha}^{1,1} \wedge \bm{\alpha}^{2,1}} x_{\bm{\alpha}^{1,1} \vee \bm{\alpha}^{2,1}} x_{\bm{\alpha}^{3,1}} -\NLM(S^0)}_{\displaystyle=S^2}\Big)
$$ 
with 
$$ 
\uglate{\bm{\alpha}^{1,1} \wedge \bm{\alpha}^{2,1}}_k \leq \uglate{\bm{\alpha}^{3,1}}_k , \uglate{\bm{\alpha}^{1,1} \vee \bm{\alpha}^{2,1}}_k \quad \text{ for all } k \in \uglate{d}, \text{ and } x_{\bm{\alpha}^{3,1}}=x_{\bm{\alpha}^{1,0} \vee \bm{\alpha}^{2,0}}.
$$ 
Next, if $S^2 \neq 0$ and  
$\LM(S^2)=x_{\bm{\alpha}^{1,1} \wedge \bm{\alpha}^{2,1}} x_{\bm{\alpha}^{1,1} \vee \bm{\alpha}^{2,1}} x_{\bm{\alpha}^{3,1}}$ then a pair of indices satisfying line 7 of Algorithm~\ref{Division}  
must be $\bm{\alpha}^{1,1} \vee \bm{\alpha}^{2,1},\bm{\alpha}^{3,1}$ 
 so that the iteration ends up with 
$$
S^2=\LC(S^2)\Big(x_{\bm{\alpha}^{3,2}} \cdot f_{d}^{\zag{\bm{\alpha}^{1,2},\bm{\alpha}^{2,2}}}+\underbrace{x_{\bm{\alpha}^{1,2} \wedge \bm{\alpha}^{2,2}} x_{\bm{\alpha}^{1,2} \vee \bm{\alpha}^{2,2}}x_{\bm{\alpha}^{3,2}}-\NLM(S^0)}_{\displaystyle=S^3}\Big)
$$ 
such that 
$$ 
\uglate{\bm{\alpha}^{3,2}}_k \leq \uglate{\bm{\alpha}^{1,2} \wedge \bm{\alpha}^{2,2}}_k \leq \uglate{\bm{\alpha}^{1,2} \vee \bm{\alpha}^{2,2}}_k
\quad \text{ for all } k \in \uglate{d}, \text{ and } x_{\bm{\alpha}^{3,2}}=x_{\bm{\alpha}^{1,1}\wedge \bm{\alpha}^{2,1}}. 
$$ 
Thus, in iteration $i=3$ the leading monomial $\LM(S^3)$ must be $\NLM(S^0)$ (unless $S^3=0$). 

A similar analysis can be performed on the monomial 
$\NLM(S^0)$ and therefore the algorithm stops after at most $6$ iterations. 
The division algorithm results in 
$$
S(f_1,f_2)=\sum_{i=0}^{p} \zag{\prod_{j=0}^i \LC(S^j)} x_{\bm{\alpha}^{3,i}}\cdot f_{d}^{\zag{\bm{\alpha}^{1,i}, \bm{\alpha}^{2,i}}},
$$ 
where $f_{d}^{\zag{\bm{\alpha}^{1,i}, \bm{\alpha}^{2,i}}}=-x_{\bm{\alpha}^{1,i} \wedge \bm{\alpha}^{2,i}}x_{\bm{\alpha}^{1,i} \vee \bm{\alpha}^{2,i}}+x_{\bm{\alpha}^{1,i}}x_{\bm{\alpha}^{2,i}} \in \mathcal{G}_d$ and $p \leq 5$.
All the cases that we left out above are treated in a similar way.
This shows that $\mathcal{G}_d$ is a Gr{\"o}bner basis of ${J}_d$. 

In order to show that $\mathcal{G}_d$ is the {\it reduced} Gr{\"o}bner basis of ${J}_d$, first notice that $\LC(g)=1$ for all $g \in \mathcal{G}_d$. 
Furthermore, the leading term of any polynomial in  $\mathcal{G}_d$ is of degree two. Thus, it is enough to show that for every pair of different polynomials $f_{d}^{(\bm{\alpha}^1,\bm{\beta}^1)},f_{d}^{(\bm{\alpha}^2,\bm{\beta}^2)} \in \mathcal{G}_d$ (related to $\mathcal{S}_1, \mathcal{M}_1$ and $\mathcal{S}_2,\mathcal{M}_2$, respectively) it holds that $\LM(f_{d}^{(\bm{\alpha}^1,\bm{\beta}^1)})\neq \LM(f_{d}^{(\bm{\alpha}^2,\bm{\beta}^2)})$ with $(\bm{\alpha}^k,\bm{\beta}^k) \in \mathcal{T}_d^{\mathcal{S}_k,\mathcal{M}_k}$ for $k = 1,2$. But this follows from the fact that all elements of $\mathcal{G}_d$ are different as remarked before the statement of the theorem. 
\end{proof}

We define the tensor $\theta_k$-norm analogously to the matrix scenario. 
\begin{definition} \label{def:theta}
The tensor \textit{$\theta_k$-norm}, denoted by $\norma{\cdot}_{\theta_k}$, is the norm induced by the $k$-theta body $\TB_k\zag{J_d}$, i.e.,
\begin{equation*}
\norma{\mbf{X}}_{\theta_k}=\inf \skup{r: \mbf{X} \in r\TB_k\zag{J_d}}.
\end{equation*}
\end{definition}
The $\theta_k$-norm can be computed with the help of Theorem~\ref{closureQ}, i.e., as 
\[
\norma{\mbf{X}}_{\theta_k} = \min t \quad \mbox{ subject to } \mbf{X} \in t \mbf{Q}_{\mathcal{B}_k}(J_d).
\]
Given the moment matrix $\mbf{M}_{\mathcal{B}_k}[\mbf{y}]$ associated with $J_d$, this minimization program is equivalent to the semidefinite program
\begin{equation}\label{matrix:sdp}
\min_{t \in \R , \mbf{y} \in \R^{\mathcal{B}_k}} t \quad \mbox{ subject to } \quad  \mbf{M}_{\mathcal{B}_k}[\mbf{y}] \succcurlyeq 0, y_0 = t, \mbf{y}_{\mathcal{B}_1} = \mbf{X}.
\end{equation}

We have focused on the polynomial ideal generated by all second order minors of all matricizations of the tensor.
One may also consider a subset of all possible matricizations corresponding to various tensor decompositions and notions of tensor rank.
For example, the Tucker(HOSVD)-rank (corresponding to the Tucker or HOSVD decomposition) of a $d$th-order tensor $\mbf{X}$ is a $d$-dimensional vector $\mbf{r}_{HOSVD}=(r_1,r_2,\ldots,r_d)$ such that $r_i=\rank\zag{\mbf{X}^{\{i\}}}$ for all $i \in \uglate{d}$, see \cite{grasedyck2010hierarchical}.  Thus, we can define an ideal $J_{d,\text{HOSVD}}$ generated by all second order minors of unfoldings $\mbf{X}^{\{k\}}$, for $k \in \uglate{d}$.

The tensor train (TT) decomposition is another popular approach for tensor computations. 
The corresponding TT-rank of a $d$th-order tensor $\mbf{X}$ is a $(d-1)$-dimensional 
vector $\mbf{r}_{TT}=(r_1,r_2,\ldots,r_{d-1})$ such that $r_i=\rank\zag{\mbf{X}^{\{1,\ldots,i\}}}$, $i \in \uglate{d-1}$, see  \cite{os11} for details. 
By taking into account only minors of order two of the matricizations $\bm{\tau} \in \skup{\{1\},\{1,2\},\ldots,\{1,2,\ldots,d-1\}}$, one may
introduce a corresponding polynomial ideal $J_{d,\text{TT}}$.

\begin{theorem}
The polynomial ideals $J_d$, $J_{d,\text{HOSVD}}$, and $J_{d,\text{TT}}$ are equal, for all $d \geq 3$.
\end{theorem}
\begin{proof}
Let $\bm{\tau} \subset \uglate{d}$ represent a matricization. Similarly to the case of order-three tensors, for  $\zag{\bm{\alpha},\bm{\beta}} \in \N^{2d}$, $x_{\bm{\alpha}^{\bm{\tau}}}x_{\bm{\beta}^{\bm{\tau}}}$ denotes the  monomial where $\alpha_k^{\bm{\tau}}=\alpha_k$, $\beta_k^{\bm{\tau}}=\beta_k$  for all $k \in \bm{\tau}$ and 
$\alpha_{\ell}^{\bm{\tau}}=\beta_{\ell}$, $\beta_{\ell}^{\bm{\tau}}=\alpha_{\ell}$  for all $\ell \in \bm{\tau}^c=\uglate{d}\backslash \bm{\tau}$. Moreover, $x_{\bm{\alpha}^{\bm{\tau},\mbf{0}}}x_{\bm{\beta}^{\bm{\tau},\mbf{0}}}$ denotes the  monomial where $\alpha_k^{\bm{\tau},\mbf{0}}=\alpha_k$, $\beta_k^{\bm{\tau},\mbf{0}}=\beta_k$  for all $k \in \bm{\tau}$ and 
$\alpha_{\ell}^{\bm{\tau},\mbf{0}}=\beta_{\ell}^{\bm{\tau},\mbf{0}}=0$  for all $\ell \in \bm{\tau}^c=\uglate{d}\backslash \bm{\tau}$.
The corresponding order-two minors are defined as
\begin{equation*}
f_{(\bm{\alpha},\bm{\beta})}^{\bm{\tau}}(\mbf{x})=-x_{\bm{\alpha}} x_{\bm{\beta}} + x_{\bm{\alpha}^{\bm{\tau}}} x_{\bm{\beta}^{\bm{\tau}}}, \quad  (\bm{\alpha},\bm{\beta}) \in\bm{\mathcal{T}}^{\bm{\tau}}.
\end{equation*}
We define the set $\bm{\mathcal{T}}^{\bm{\tau}}$ as
$$\bm{\mathcal{T}}^{\bm{\tau}}= \skup{\zag{\bm{\alpha},\bm{\beta}}: \bm{\alpha}^{\bm{\tau},\mbf{0}}\neq\bm{\beta}^{\bm{\tau},\mbf{0}}, \, \bm{\alpha}^{\bm{\tau}^c,\mbf{0}}\neq\bm{\beta}^{\bm{\tau}^c,\mbf{0}}}.$$
Similarly as in the case of order-three tensors, notice that $f_{(\bm{\alpha},\bm{\beta})}^{\bm{\tau}}(\mbf{x})= f_{(\bm{\beta},\bm{\alpha})}^{\bm{\tau}}(\mbf{x})=-f_{(\bm{\alpha}^{\bm{\tau}},\bm{\beta}^{\bm{\tau}})}^{\bm{\tau}}(\mbf{x})=-f_{(\bm{\beta}^{\bm{\tau}},\bm{\alpha}^{\bm{\tau}})}^{\bm{\tau}}(\mbf{x})$, for all $\zag{\bm{\alpha},\bm{\beta}} \in \bm{\mathcal{T}}^{\bm{\tau}}$.
First, we show that $J_d=J_{d,\text{HOSVD}}$ by showing that $f_{(\bm{\alpha},\bm{\beta})}^{\bm{\tau}}(\mbf{x}) \in J_{d,\text{HOSVD}}$, for all $(\bm{\alpha},\bm{\beta}) \in \bm{\mathcal{T}}^{\bm{\tau}}$ and all $\aps{\bm{\tau}}\geq 2$. Without loss of generality, we can assume that $\alpha_i \neq \beta_i$, for all $i \in \bm{\tau}$ since otherwise we can consider the matricization $\bm{\tau}\backslash\skup{i: \alpha_i = \beta_i}$. Additionally, by definition of $\bm{\mathcal{T}}^{\bm{\tau}}$, there exists at least one $\ell \in \bm{\tau}^c$ such that $\alpha_{\ell} \neq \beta_{\ell}$.
Let $\bm{\tau}=\{t_1,t_2,\ldots,t_k\}$ with $t_i<t_{i+1}$, for all $i \in \uglate{k-1}$ and $k\geq 2$. Next, fix $\zag{\bm{\alpha},\bm{\beta}} \in \mathcal{T}^{\bm{\tau}}$ and define $\bm{\alpha}^0=\bm{\alpha}$ and $\bm{\beta^0}=\bm{\beta}$. Algorithm \ref{JdTT} results in polynomials $g_k \in J_{3,\text{TT}}$ such that $f_{(\bm{\alpha},\bm{\beta})}^{\bm{\tau}}(\mbf{x})=\sum_{i=1}^k g_i(\mbf{x})$. This follows from
$$\sum_{i=1}^k g_i = \sum_{i=1}^k \zag{-x_{\bm{\alpha}^{i-1}} x_{\bm{\beta}^{i-1}}+x_{\bm{\alpha}^{i}} x_{\bm{\beta}^{i}}}= -x_{\bm{\alpha}^{0}} x_{\bm{\beta}^{0}} + x_{\bm{\alpha}^{k}} x_{\bm{\beta}^{k}}=f_{(\bm{\alpha},\bm{\beta})}^{\bm{\tau}}(\mbf{x}).  $$
By the definition of polynomials $g_k$ it is obvious that  $$g_i \in \skup{f_{(\bm{\alpha},\bm{\beta})}^{\{i\}}(\mbf{x}): \zag{\bm{\alpha},\bm{\beta}} \in \bm{\mathcal{T}}^{\{i\}}}, \text{ for all } i \in\uglate{k}.$$
\begin{algorithm}
\caption{\normalsize Algorithm for proving that $J_d=J_{d,\text{TT}}$}
\label{JdTT}
\begin{algorithmic}
\REQUIRE  An ideal $J_{d,\text{TT}} \in \R\uglate{\mbf{x}}$, polynomial $f_{(\bm{\alpha},\bm{\beta})}^{\bm{\tau}}(\mbf{x})$ with $\bm{\alpha}^0=\bm{\alpha}, \bm{\beta}^0=\bm{\beta},\bm{\tau}=\{t_1,t_2,\ldots,t_k\}$, where $k\geq 2$ 
\FOR{$i=1,\ldots,k$} 
\STATE \textbf{Define} $\bm{\alpha}^i$ \textbf{and} $\bm{\beta}^i$ \textbf{as} $$\alpha_j^i:=\begin{cases} 
\beta_j^{i-1} & \mbox{if } j = t_i, \\
\alpha_j^{i-1} & \text{otherwise} 
\end{cases} \quad  \text{and} \quad \beta_j^i:=\begin{cases} 
\alpha_j^{i-1} & \mbox{if } j = t_i, \\
\beta_j^{i-1} & \text{otherwise}.
\end{cases}
$$
\STATE \textbf{Define polynomial} $g_i(\mbf{x}):=-x_{\bm{\alpha}^{i-1}} x_{\bm{\beta}^{i-1}} + x_{\bm{\alpha}^{i}}x_{\bm{\beta}^{i}}$ . 
\ENDFOR
\ENSURE Polynomials $g_1,g_2,\ldots,g_k$.
\medskip
\end{algorithmic}
\end{algorithm}
Next, we show that $J_d=J_{d,\text{TT}}$. Since $J_d=J_{d,\text{HOSVD}}$, it is enough to show that $f_{(\bm{\alpha},\bm{\beta})}^{\{k\}} \in J_{d,\text{TT}}$, for all  $(\bm{\alpha},\bm{\beta}) \in \bm{\mathcal{T}}^{\{k\}}$ and all $k \in \uglate{d}$. By definition of $J_{d,\text{TT}}$ this is true for $k=1$.  Fix $k \in \{2,3,\ldots,d\}$, $(\bm{\alpha},\bm{\beta}) \in \bm{\mathcal{T}}^{\{k\}}$ and consider a polynomial  $f(\mbf{x})=f_{(\bm{\alpha},\bm{\beta})}^{\{k\}}(\mbf{x})$ corresponding to the second order minor of the matricization $\mbf{X}^{\{k\}}$. By definition of $\bm{\mathcal{T}}^{\{k\}}$, $\alpha_k\neq \beta_k$  and there exists an index $i \in \uglate{d}\backslash\{k\}$ such that $\alpha_i \neq \beta_i$. Assume that $i>k$.
 Define the polynomials
$g(\mbf{x}) \in \bm{\mathcal{R}}^{\{1,2,\ldots,k\}}:=\skup{f_{(\bm{\alpha},\bm{\beta})}^{\{1,2,\ldots,k\}}(\mbf{x}): \zag{\bm{\alpha},\bm{\beta}} \in \bm{\mathcal{T}}^{\{1,2,\ldots,k\}}} $ and $h(\mbf{x}) \in \bm{\mathcal{R}}^{\{1,2,\ldots,k-1\}}:=\skup{f_{(\bm{\alpha},\bm{\beta})}^{\{1,2,\ldots,k-1\}}(\mbf{x}): \zag{\bm{\alpha},\bm{\beta}} \in \bm{\mathcal{T}}^{\{1,2,\ldots,k-1\}}}$  as
\begin{align*}
g(\mbf{x})& = -x_{\bm{\alpha}} x_{\bm{\beta}} + x_{\bm{\alpha}^{\{1,2,\ldots,k\}}} x_{\bm{\beta}^{\{1,2,\ldots,k\}}} \\
h(\mbf{x})& = - x_{\bm{\alpha}^{\{1,2,\ldots,k\}}} x_{\bm{\beta}^{\{1,2,\ldots,k\}}} +  x_{{\bm{\alpha}^{\{1,2,\ldots,k\}}}^{\{1,2,\ldots,k-1\}}} x_{{\bm{\beta}^{\{1,2,\ldots,k\}}}^{\{1,2,\ldots,k-1\}}} 
\end{align*}
Since  $ x_{{\bm{\alpha}^{\{1,2,\ldots,k\}}}^{\{1,2,\ldots,k-1\}}} x_{{\bm{\beta}^{\{1,2,\ldots,k\}}}^{\{1,2,\ldots,k-1\}}}=x_{\bm{\alpha}^{\{k\}}} x_{\bm{\beta}^{\{k\}}}$,  we have $f(\mbf{x})=g(\mbf{x})+h(\mbf{x})$ and thus $f \in J_{d,\text{TT}}$. If $i<k$ notice that $f(\mbf{x})=g_1(\mbf{x})+h_1(\mbf{x})$, where
\begin{align*}
g_1(\mbf{x})& = -x_{\bm{\alpha}} x_{\bm{\beta}} + x_{\bm{\alpha}^{\{1,2,\ldots,k-1\}}} x_{\bm{\beta}^{\{1,2,\ldots,k-1\}}} \in \bm{\mathcal{R}}^{\{1,2,\ldots,k-1\}}\\
h_1(\mbf{x})& = - x_{\bm{\alpha}^{\{1,2,\ldots,k-1\}}} x_{\bm{\beta}^{\{1,2,\ldots,k-1\}}} +  x_{{\bm{\alpha}^{\{1,2,\ldots,k-1\}}}^{\{1,2,\ldots,k\}}} x_{{\bm{\beta}^{\{1,2,\ldots,k-1\}}}^{\{1,2,\ldots,k\}}}\\
 &=-x_{\bm{\alpha}^{\{1,2,\ldots,k\}}} x_{\bm{\beta}^{\{1,2,\ldots,k\}}} + x_{\bm{\alpha}^{\{k\}}} x_{\bm{\beta}^{\{k\}}} \in \bm{\mathcal{R}}^{\{1,2,\ldots,k\}}. 
\end{align*}
\end{proof}
\begin{remark}
Fix a decomposition tree $T_I$ which generates a particular HT-decomposition and consider the ideal $J_{d,\text{HT},T_I}$ generated by all second order minors corresponding to the matricizations induced by the tree $T_I$. In a similar way as above, one can obtain that  $J_{d,\text{HT},T_I}$ equals to $J_{d}$.
\end{remark}

\section{Convergence of the unit-$\theta_k$-norm balls} \label{Section:Convergence}
In this section we show the following result on the convergence of the unit $\theta_k$-balls . 
\begin{theorem}\label{THconvergence}
The theta body sequence of $J_d$ converges asymptotically to the $\conv\zag{\nu_{\R}(J)}$, i.e., 
$$ \bigcap_{k=1}^{\infty} \TB_k(J_d)=\conv\zag{\nu_{\R}(J_d)}.$$
\end{theorem}

To prove Theorem \ref{THconvergence} we use the following result presented in \cite{blpath13} which is a consequence of Schm{\"u}dgen's Positivstellensatz.
\begin{theorem}\label{ThInfConverg}
Let $J$ be an ideal such that $\nu_{\R}(J)$ is compact. Then the theta body sequence of $J$ converges to the convex hull of the variety $\nu_{\R}(J)$, in the sense that
$$ \bigcap_{k=1}^{\infty} \TB_k(J)=\conv\zag{\nu_{\R}(J)}.$$
\end{theorem}
\begin{proof}[Proof of Theorem \ref{THconvergence}]
The set $\nu_{\R}(J_d)$ is the set of rank-one tensors with unit Frobenius norm which can be written as $\nu_{\R}(J_d)=\mathcal{A}_1 \bigcap \mathcal{A}_2$ where
\begin{align*}
\mathcal{A}_1&=\skup{\mbf{X} \in \Rd: \rank(\mbf{X})=1}, \\
\text{and }\mathcal{A}_2&=\skup{\mbf{X} \in \Rd: \norma{\mbf{X}}_F=1}.
\end{align*}
It is well-known that $\mathcal{A}_1$ is closed \cite[discussion before Definition 2.2]{cartwright2012secant} and since $\mathcal{A}_2$ is clearly compact, $\nu_{\R}(J_d)$ is compact. Therefore, the result follows from Theorem \ref{ThInfConverg}.
\end{proof}

\section{Computational Complexity}
\label{sec:complexity}

The computational complexity of the semidefinite programs for computing the $\theta_1$-norm of a tensor or for minimizing the $\theta_1$-norm
subject to a linear constraint depends polynomially on the number of variables, i.e., on the size of ${\mathcal B}_{2k}$, and on
the dimension of the moment matrix $\mbf{M}$. We claim that the overall complexity scales polynomially in $n$, 
where for simplicity we consider $d$th-order tensors in $\R^{n \times n \times \cdots \times n}$. Therefore, in contrast to tensor nuclear norm minimization which is NP-hard for $d \geq 3$, tensor recovery via $\theta_1$-norm minimization is tractable.

Indeed, the moment matrix $\mbf{M}$ is of dimension $(1+n^d) \times (1+n^d)$ (see also \eqref{MatrixM1} for matrices in $\R^{2 \times 2}$) and
if $a = n^d$ denotes the total number of entries of a tensor $\mbf{X} \in \R^{n \times \cdots \times n}$, 
then the number of the variables is at most $\frac{a\cdot (a+1)}{2} \sim \mathcal{O}(a^2)$ which is polynomial in $a$.
(A more precise counting does not give a substantially better estimate.)

\ifnocut

\fi

\section{Numerical experiments}
\label{sec:numerics}

Let us now empirically study the performance of low rank tensor recovery via $\theta_1$-norm minimization via numerical experiments, where
we concentrate on third-order tensors. Given measurements $\mathbf{b} = \mathbf{\Phi}(\mathbf{X})$ of a low rank tensor 
$\mbf{X} \in \R^{n_1 \times n_2 \times n_3}$,
where $\mathbf{\Phi} : \R^{n_1 \times n_2 \times n_3} \to \R^m$ is a linear measurement map, we aim at reconstructing
$\mathbf{X}$ as the solution of the minimization program
\begin{equation}\label{theta1:min:prog}
\min \| \mathbf{Z} \|_{\theta_1} \quad \mbox{ subject to } \mathbf{\Phi}(\mathbf{Z}) = \mathbf{b}.
\end{equation}
As outlined in Section~\ref{IntrodTheta}, the $\theta_1$-norm of a tensor $\mathbf{Z}$ can be computed as the minimizer of
the semidefinite program
\[
\min_{t, \mathbf{y}} t \quad \mbox{ subject to }  \quad \mathbf{M}(t,\mathbf{y}, \mathbf{Z}) \succcurlyeq 0,
\]
where $\mathbf{M}(t,\mathbf{y},\mathbf{X}) = \mathbf{M}_{\mathcal{B}_1}(t, \mathbf{X}, \mathbf{y})$ is the moment matrix of order $1$ 
associated to the ideal $J_3$, see Theorem~\ref{th3rd}. This moment matrix for $J_3$ is explicitly given by 
\begin{equation*}
\mbf{M}\zag{t,\mbf{y},\mbf{X}}=t\mbf{M}_0+\sum_{i=1}^{n_1}\sum_{j=1}^{n_2}\sum_{k=1}^{n_3} X_{ijk}\mbf{M}_{ijk}+\sum_{p=2}^9\sum_{q=1}^{\aps{\mbf{M}^p}}y_{\ell}\mbf{M}_{h_p(q)}^p,
\end{equation*}
where $\ell=\sum_{r=2}^{p-1}\aps{\mbf{M}^{r}}+q$, $\mbf{M}^{p}=\{\mbf{M}_{\widetilde{I}}^{p}\}$, and 
the matrices $\mbf{M}_0, \mbf{M}_{ijk}$ and $\mbf{M}_{\widetilde{I}}^p$ are provided in Table~\ref{TableM}. 
For $p \in \{2,3,\ldots,9\}$,  the function $h_p$ denotes an arbitrary but fixed bijection $
\skup{1,2,\ldots,\aps{\mbf{M}^{p}}} \mapsto \{(i,\hat{i},j,\hat{j},k,\hat{k})\}$, where  $\widetilde{I}=(i,\hat{i},j,\hat{j},k,\hat{k})$ is in the range of the last column of Table~\ref{TableM}. 
As discussed in Section~\ref{IntrodTheta} for the general case, 
the $\theta_1$-norm minimization problem \eqref{theta1:min:prog} is then equivalent to the semidefinite program  \begin{equation}\label{theta1:semidefinite}
\min_{t, \mbf{y}, \mbf{Z}} t \quad \mbox{ subject to }\quad \mbf{M}\zag{t, \mbf{y}, \mbf{Z}}\succeq 0 \quad \mbox{ and  } \quad
\mathbf{\Phi}(\mbf{Z}) = \mbf{b}.
\end{equation}

\begin{table}
\caption{The matrices involved in the definition of the moment matrix $\mbf{M}\zag{t,\mbf{y},\mbf{X}}$. Due to the symmetry only the upper triangle part of the matrices is specified. The other non-specified entries of the matrices $\mbf{M}\in \R^{(n_1n_2n_3+1) \times (n_1n_2n_3+1)}$ from the first column are equal to zero.  
The matrix $\mbf{M}$ corresponds to the element $g+{J}_3$ of the $\theta$-basis  specified in the second column. The index $\widetilde{I}=(i,\hat{i},j,\hat{j},k,\hat{k})$ is in the range of the last column.
The function $f:\N^3\rightarrow \N$  is defined as $f\zag{i,j,k}=(i-1)n_2n_3+(j-1)n_3+k+1$. }
\label{TableM}
\begin{tabular}{l l  l  r  l }
\hline\noalign{\smallskip}
&$\theta$-basis & position $(p,q)$ in the matrix & ${M}_{pq}$ & Range of $i,\hat{i},j,\hat{j},k,\hat{k}$ \\
\noalign{\smallskip}\hline\noalign{\smallskip}
$\mbf{M}_0$& $1$ &  $\zag{1,1},\zag{2,2}$ & $1$ & \\
$\mbf{M}_{ijk}$& $x_{ijk}$&  $\zag{1,f(i,j,k)}$ & $1$ & $i \in \uglate{n_1}, j \in \uglate{n_2}, k \in \uglate{n_3}$ \\
$\mbf{M}_{f_2}^2$&${x}_{ijk}^2$ & $\zag{2,2}$ &  $-1$  &\\
&                         & $\zag{f(i,j,k),f(i,j,k)}$ &  $1$ & $\skup{i \in \uglate{n_1}, j \in \uglate{n_2}, k \in \uglate{n_3}}$ \\  & & & & \hspace{2em} $\backslash \skup{i=j=k=1}$ \\
$\mbf{M}_{f_3}^3$&${x}_{i\hat{j}k}x_{ij\hat{k}}$ & $(f(i,j,k),f(i,\hat{j}, \hat{k})),$   &$1$ & \\
& & $(f(i,j,\hat{k}),f(i,\hat{j}, k))$ & $1$ & $i \in \uglate{n_1}, j < \hat{j}, k < \hat{k}$\\
$\mbf{M}_{f_4}^4$&$x_{ijk}{x}_{\hat{i}\hat{j}\hat{k}}$ & $(f(i,j,k),f(\hat{i},\hat{j}, \hat{k}))$   &$1$ &\\
& & $(f(i,\hat{j},k),f(\hat{i},j,\hat{k}))$   &$1$ &\\
&								   & $(f(i,\hat{j},\hat{k}), f(\hat{i},j,k)),$	& $1$	&\\
&								   & $(f(i,j,\hat{k}),f( \hat{i},\hat{j},k))$	& $1$					& $i < \hat{i}, j < \hat{j}, k < \hat{k}$\\
$\mbf{M}_{f_5}^5$&$x_{ijk}{x}_{\hat{i}j\hat{k}}$ & $(f(i,j,k),f(\hat{i}, j, \hat{k})),$   &$1$ &\\		   
& & $(f(i,j,\hat{k}),f(\hat{i},j, k))$   &$1$ & $i < \hat{i}, j \in \uglate{n_2}, k < \hat{k}$\\	
$\mbf{M}_{f_6}^6$&$x_{ijk}{x}_{\hat{i}\hat{j}k}$ & $(f(i,j,k),f(\hat{i},\hat{j},k))$   &$1$ & \\
& & $ (f(i,\hat{j},k),f(\hat{i},j, k))$   &$1$ & $i < \hat{i}, j < \hat{j}, k \in\uglate{n_3}$\\
$\mbf{M}_{f_7}^7$&${x}_{\hat{i}jk}x_{ijk}$ & $(f(i,j,k),f(\hat{i},j,k))$   &$1$ & $i  < \hat{i}, j \in \uglate{n_2}, k \in \uglate{n_3}$\\
$\mbf{M}_{f_8}^8$&$x_{i\hat{j}k}x_{ijk}$ & $(f(i,j,k),f(i,\hat{j},k))$   &$1$ & $i  \in \uglate{n_1}, j < \hat{j}, k \in \uglate{n_3}$\\
$\mbf{M}_{f_9}^9$&$x_{ij\hat{k}}x_{ijk}$ & $(f(i,j,k),f(i,j,\hat{k}))$   &$1$ & $i  \in \uglate{n_1}, j \in \uglate{n_2}, k < \hat{k}$\\
\noalign{\smallskip}\hline
\end{tabular}
\vspace{0.3cm}
\end{table}

For our experiments, the linear mapping is defined as $(\mathbf{\Phi}\zag{\mathbf{X}})_k=\interval{\mathbf{X},\mathbf{\Phi}_k}$, $k \in [m]$, 
with independent Gaussian random tensors $\mathbf{\Phi}_k \in \R^{n_1 \times n_2 \times n_3}$, i.e., all entries of $\mathbf{\Phi}_k$ are independent
$\mathcal{N}\zag{0,\frac{1}{m}}$ random variables. We choose tensors $\mathbf{X} \in \R^{n_1 \times n_2 \times n_3}$ of rank one
as $\mathbf{X}=\mathbf{u} \otimes \mathbf{v} \otimes \mathbf{w}$, where each entry of the vectors $\mathbf{u}$, $\mathbf{v}$, and $\mathbf{w}$ is taken independently from the normal distribution $\mathcal{N}\zag{0,1}$. Tensors $\mathbf{X} \in \R^{n_1 \times n_2 \times n_3}$ of rank two
are generated as the sum of two random rank-one tensors. With $\mathbf{\Phi}$ and $\mathbf{X}$ given, we compute $\mathbf{b} = \mathbf{\Phi}(\mathbf{X})$,
run the semidefinite program \eqref{theta1:semidefinite} and compare its minimizer with the original low rank tensor $\mathbf{X}$.
For a given set of parameters, i.e., dimensions $n_1,n_2,n_3$, number of measurements $m$ and rank $r$, we repeat this experiment $200$ times
and record the empirical success rate of recovering the original tensor, where we say that recovery is successful if the elementwise
reconstruction error is at most $10^{-6}$. We use MATLAB (R2008b) for these numerical experiments, including SeDuMi\_1.3 for solving the semidefinite programs. 

Table~\ref{RecNumRes} summarizes the results of our numerical tests
for cubic and non-cubic tensors of rank one and two and several choices of the dimensions. Here, the number $m_0$ denotes the maximal number of measurements for which not even one out of $200$ generated tensors is recovered and $m_1$ denotes the minimal number of measurements for which all $200$ tensors are recovered. The fifth column in Table~\ref{RecNumRes} represents the number of independent measurements which are always sufficient for the recovery of a tensor of an arbitrary rank. For illustration, we present the average cpu time (in seconds) for solving the semidefinite programs
via SeDuMi\_1.3 in the last column. 
Alternatively, the SDPNAL+ Matlab toolbox (version 0.5 beta) for semidefinite programming \cite{sdpnal1,sdpnal2} allows to perform low rank tensor recovery via $\theta_1$-norm minimization 
for even higher-dimensional tensors. For example, with $m=95$ measurement we managed to recover all rank-one $9\times 9 \times 9$ tensors 
out of $200$ (each simulation taking about $5$min). Similarly, rank-one $11 \times 11 \times 11$ tensors are recovered from $m=125$ measurements
with one simulation lasting about $50$min. Due to these large computation times, more elaborate numerical experiments have not been conducted in these scenarios.
We remark that no attempt of accelerating the optimization algorithm has been made. This task is left for future research.

\begin{table}[!h]
\caption{Numerical results for low rank tensor recovery in $\R^{n_1 \times n_2 \times n_3}$.}
\label{RecNumRes}
	\begin{tabular}{c c c c c r}
	\hline\noalign{\smallskip}
	$n_1 \times n_2 \times n_3$ & rank & $m_0$ & $m_1$ & $n_1 n_2 n_3$ & cpu (sec)\\ 
	\noalign{\smallskip}\hline\noalign{\smallskip}
	$2 \times 2 \times 3$  & $1$ & $4$ & $12$ & $12$ & $0.2$\\
	$ 3 \times 3 \times 3$ & $1$ & $ 6$ & $19 $ & $27$ & $0.37$\\
	$3 \times 4 \times 5$  & $1$ & $11$ & $30$ & $60$ & $6.66$\\
	$4 \times 4 \times 4$  & $1$ & $11$ & $32$ & $64$ & $7.28$\\
	$4 \times 5 \times 6$  & $1$ & $18$ & $42$ & $120$ & $129.48$\\
	$5 \times 5 \times 5$  & $1$ & $18$ & $43$ & $125$ & $138.90$\\
	$ 3 \times 4 \times 5$ & $2$ & $ 27$ & $ 56$ & $60$ & $7.55$\\
	$ 4 \times 4 \times 4$ & $2$ & $ 26$ & $ 56$ & $64$ & $8.65$\\
	$ 4 \times 5 \times 6$ & $2$ & $ 41$ & $ 85$ & $120$ & $192.58$\\
	\noalign{\smallskip}\hline
	\end{tabular}
\end{table}

Except for very small tensor dimensions, we can always recover tensors of rank-one or two from a number of measurements
which is significantly smaller than the dimension of the corresponding tensor space. Therefore, low rank tensor recovery via $\theta_1$-minimization
seems to be a promising approach. Of course, it remains to investigate the recovery performance theoretically.

Figures \ref{Figure1} and \ref{Figure2} present the numerical results for low rank tensor recovery via $\theta_1$-norm minimization for Gaussian measurement maps, conducted with the SDPNAL+ toolbox. 
For fixed tensor dimensions $n \times n \times n$, fixed tensor rank $r$, and fixed number $m$ of measurements $50$ simulations are performed. 
We say that recovery is successful if the element-wise reconstruction error is smaller than $10^{-3}$. Figures \ref{1a}, \ref{2a}, \ref{3a} and \ref{1b}, \ref{2b}, \ref{3b} present experiments for rank-one and rank-two tensors, respectively.  The vertical axis in all three figures represents the empirical success rate.
In Figure \ref{Figure1} the horizontal axis represents the relative number of measurements, to be more precise, for a tensor of size $n \times n \times n$, the number $\bar{n}$ on the horizontal axis represents $m=\bar{n} \frac{n^3}{100}$ measurements.
In Figure \ref{Figure2} for a rank-$r$ tensor of size $n\times n \times n$ and the number of measurements $m$, the horizontal axis represents  the number $m/(3nr)$. Notice that $3nr$ represents the degrees of freedom in the corresponding CP-decomposition. In particular, if the number of measurements necessary for tensor recovery is $m \geq 3Crn$, for an universal constant $C$, Figure \ref{Figure2} suggests that the constant C depends on the size of the tensor. In particular, it seems to grow slightly with $n$ (although it is still possible that there exists $C>0$ such that $m\geq 3Crn$ would always be enough for the recovery). With $C=3.3$ we would always be able to recover a low rank tensor of size $n \times n \times n$ with $n\leq 7$. The horizontal axis in Figure \ref{Figure3} represents the number $m/\zag{3nr\cdot \log (n)}$. The figure suggests that with the number of measurements $m \geq 6rn\cdot\log(n)$ we would always be able to recover a low rank tensor and therefore it may be possible that a logarithmic factor is necessary.
The computation is implemented in MATLAB R2016a, on an Acer Laptop with CPU@1.90GHz and RAM 4GB.

\begin{figure}
\subfloat[Recovery of rank-$1$ tensors \label{1a}]{\includegraphics[scale=0.42,trim={110 240 100 235},clip]{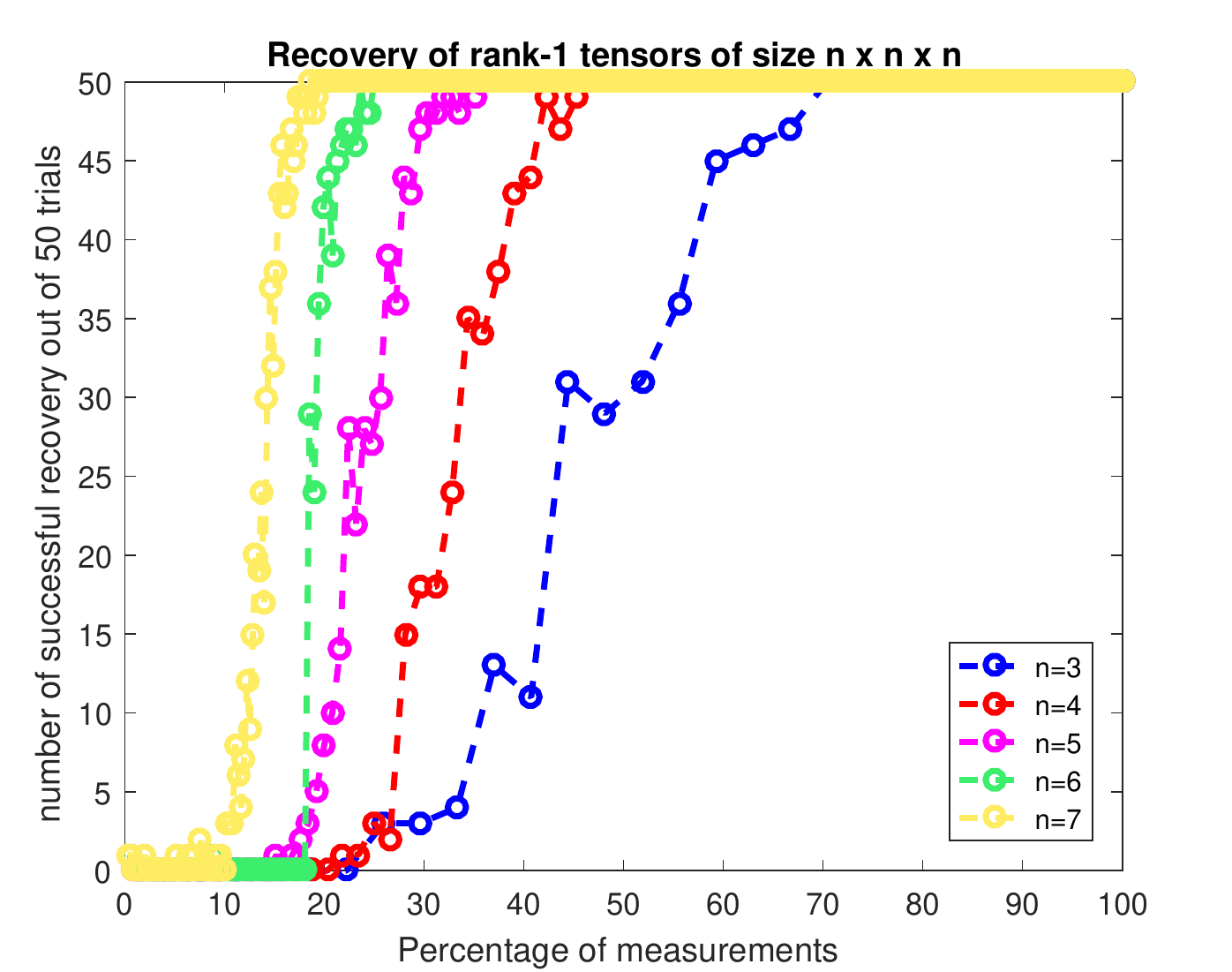}}
\subfloat[Recovery of rank-$2$ tensors \label{1b}]{\includegraphics[scale=0.42,trim={110 240 120 235},clip]{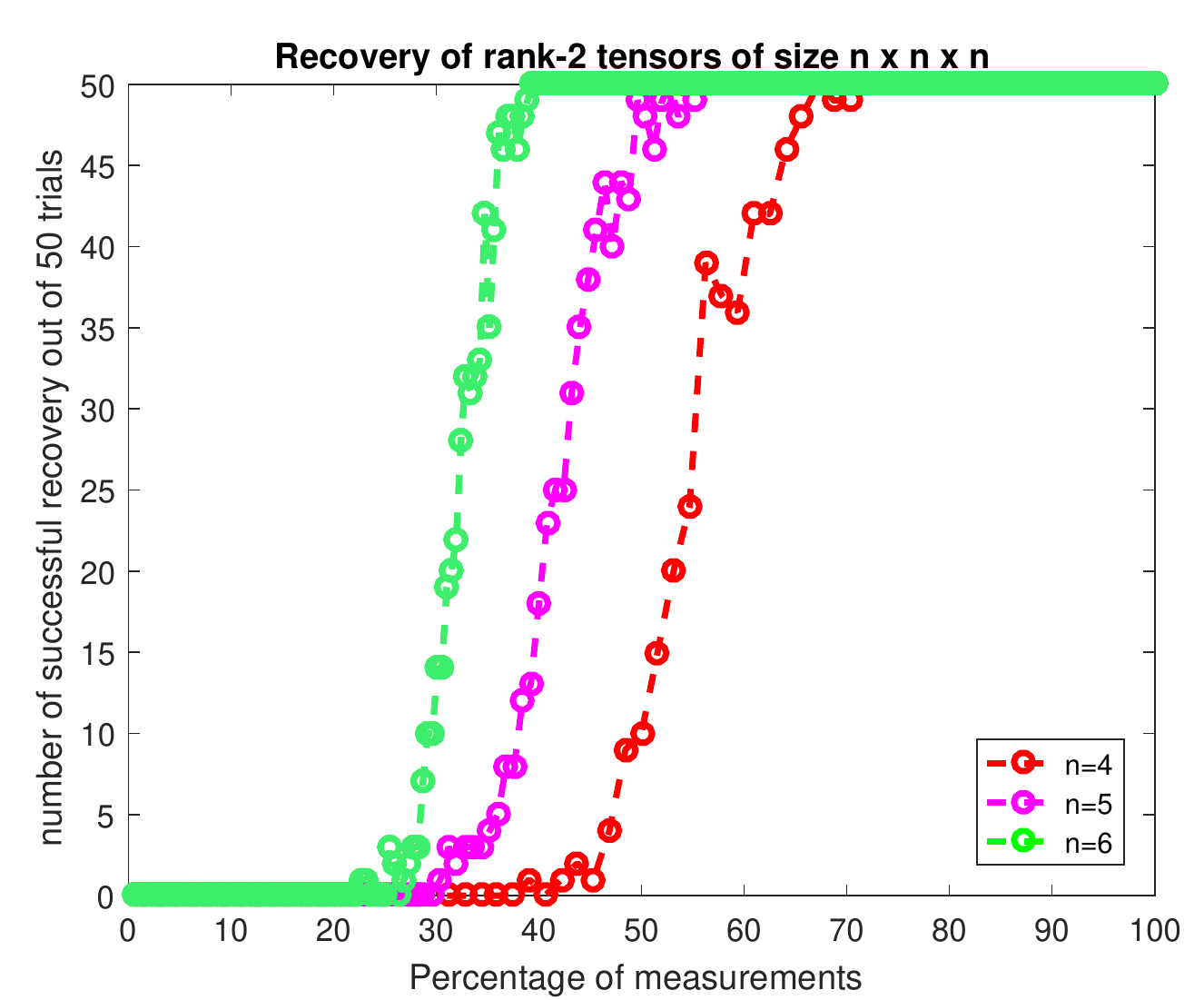}}
   \caption{Recovery of rank-$1$ and rank-$2$ tensors via $\theta_1$-norm minimization.}
  \label{Figure1}
\end{figure}

\begin{figure}
\subfloat[Recovery of rank-$1$ tensors \label{2a}]{\includegraphics[scale=0.42,trim={110 240 100 235},clip]{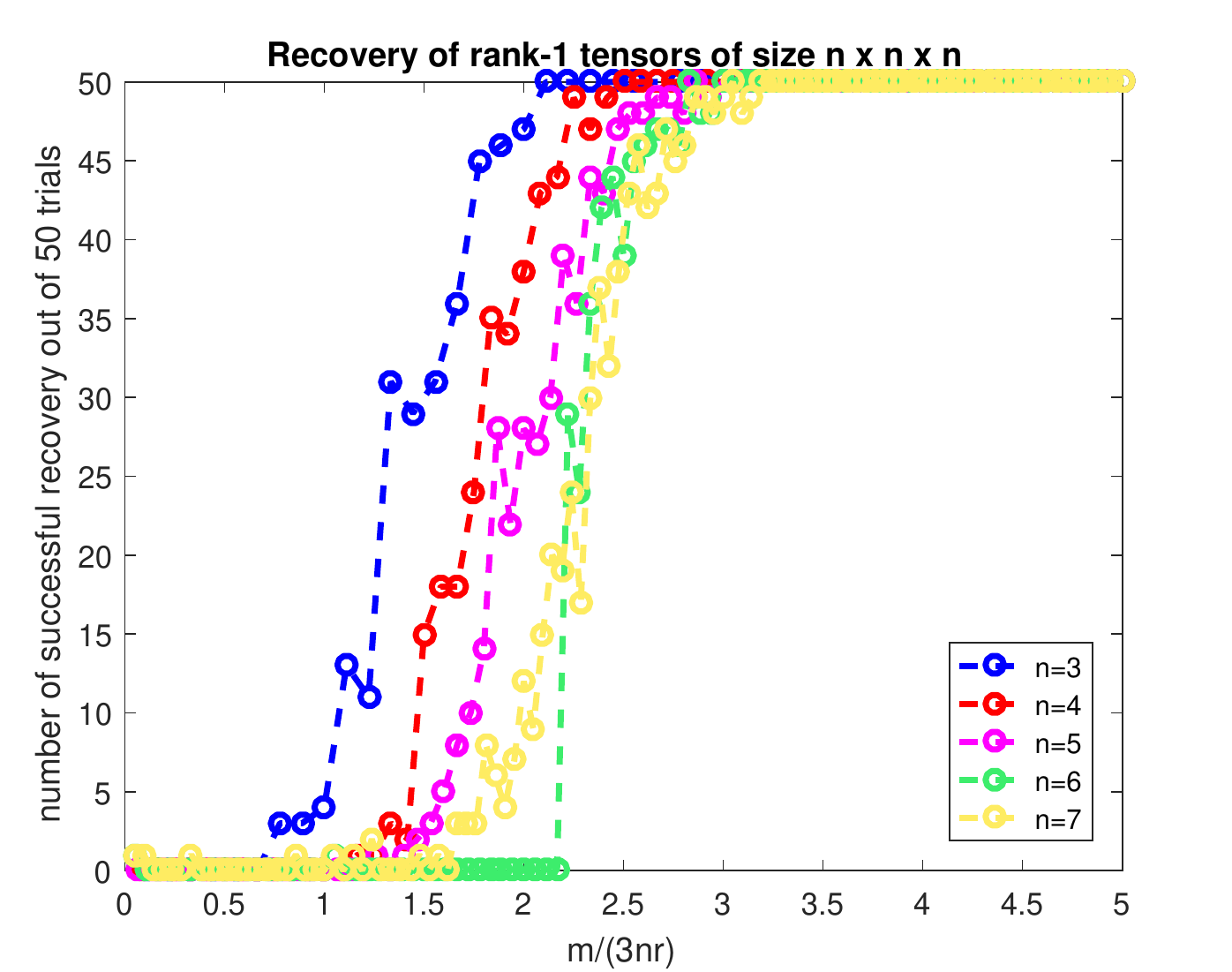}}
\subfloat[Recovery of rank-$2$ tensors \label{2b}]{\includegraphics[scale=0.42,trim={110 240 120 235},clip]{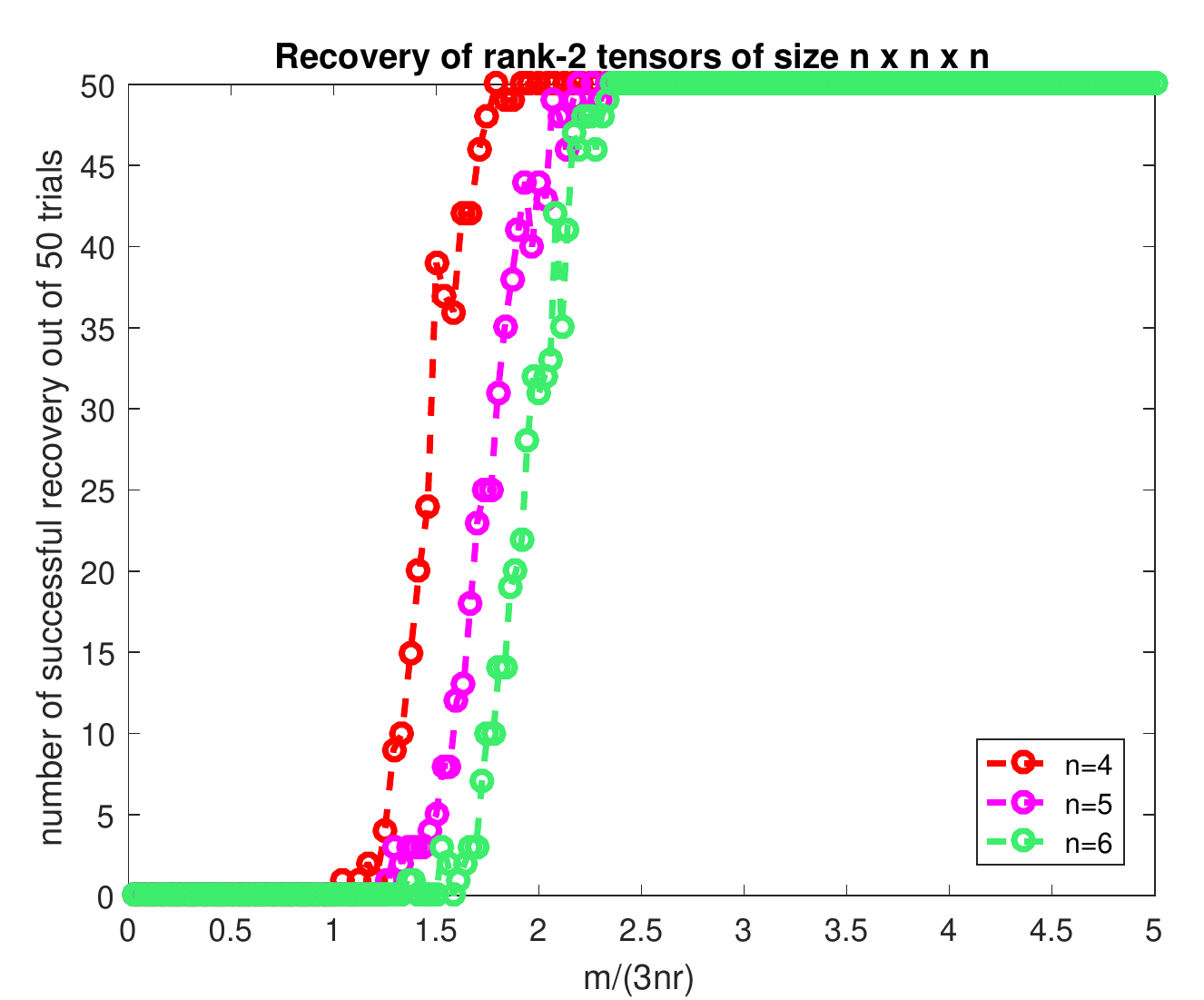}}
   \caption{Recovery of rank-$1$ and rank-$2$ tensors via $\theta_1$-norm minimization.}
  \label{Figure2}
\end{figure}

\begin{figure}
\subfloat[Recovery of rank-$1$ tensors \label{3a}]{\includegraphics[scale=0.42,trim={110 240 100 235},clip]{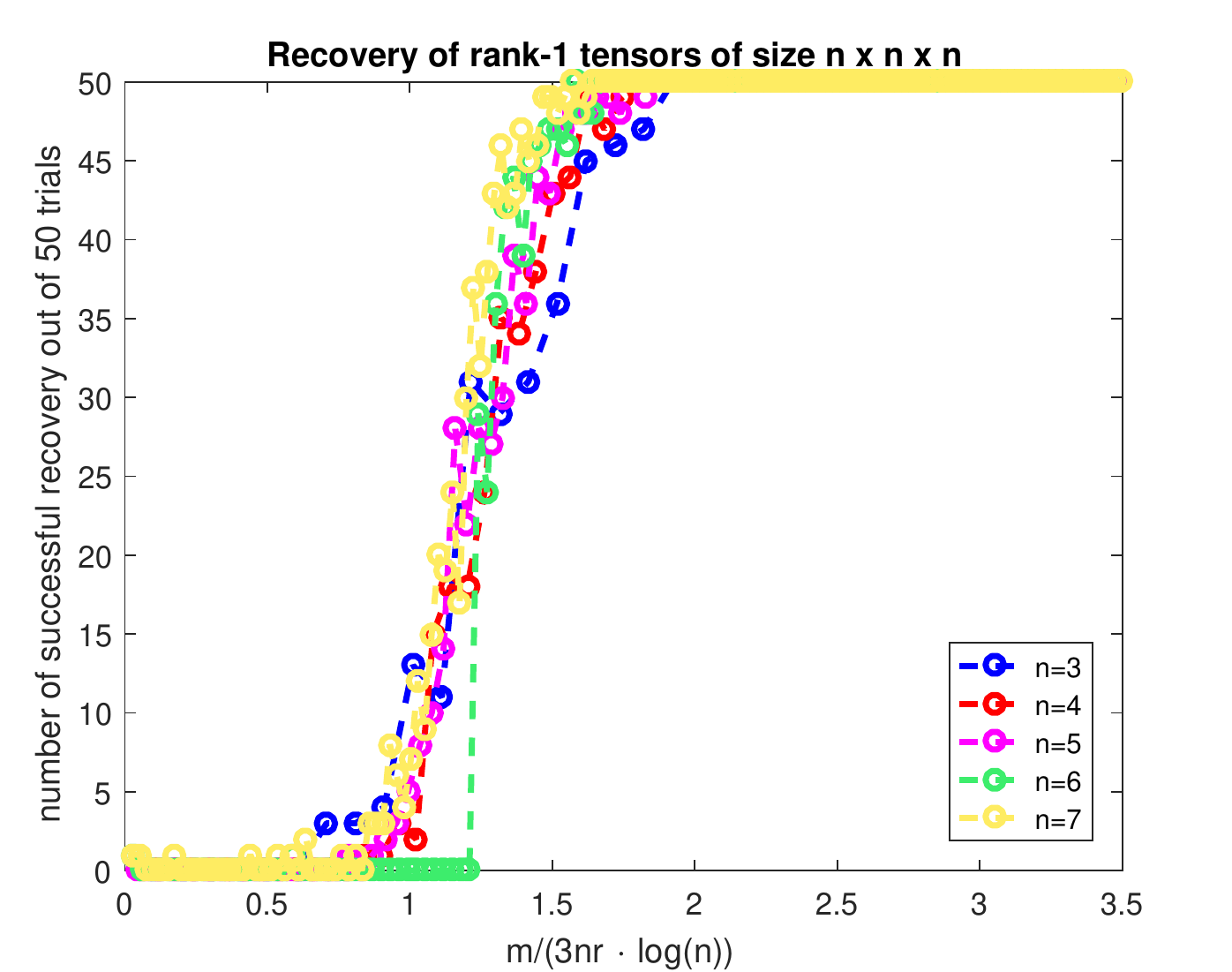}}
\subfloat[Recovery of rank-$2$ tensors \label{3b}]{\includegraphics[scale=0.42,trim={110 240 120 235},clip]{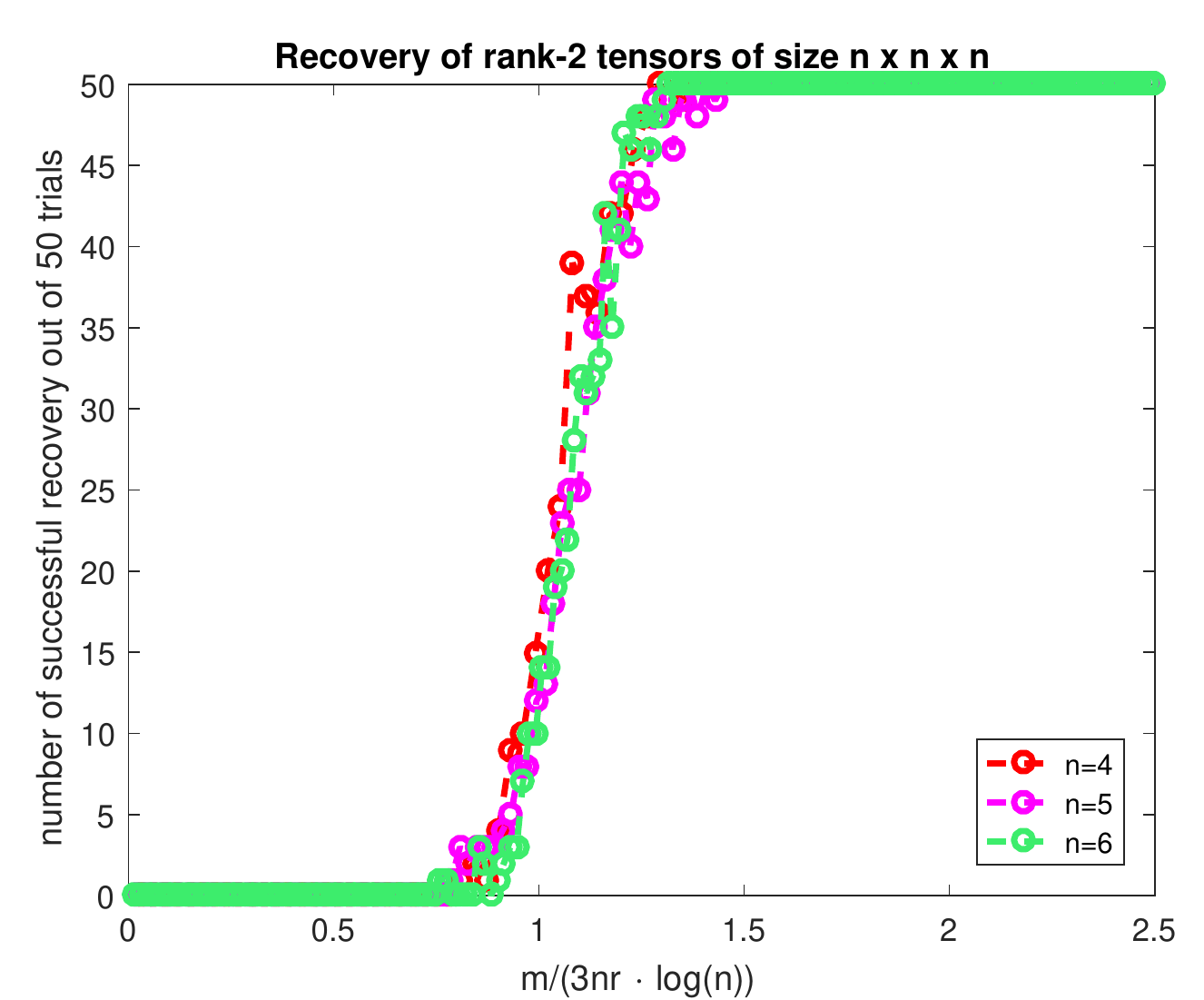}}
   \caption{Recovery of rank-$1$ and rank-$2$ tensors via $\theta_1$-norm minimization.}
  \label{Figure3}
\end{figure}

We remark that we have used standard MATLAB packages for convex optimization to perform the numerical experiments. To obtain better performance, new optimization methods  should be developed specifically to solve our optimization problem, or more generally, to solve the sum-of-squares polynomial problems. We expect this to be possible and the resulting algorithms to give much better performance results since we have shown that in the matrix scenario all theta norms correspond to the matrix nuclear norm. The state-of-the-art algorithms developed for the matrix scenario can compute the matrix nuclear norm and can solve the matrix nuclear norm minimization problem for matrices of large dimensions. The theory developed in this paper together with the first numerical results should encourage the development into this direction.

\appendix
\section{Monomial orderings and Gr{\"o}bner bases}
\label{sec:appA}

An ordering on the set of monomials $\mbf{x}^\alpha \in \R[\mbf{x}]$, $\mbf{x}^{\alpha} = x_1^{\alpha_1}\cdot x_2^{\alpha_2} \cdots x_n^{\alpha_n}$, is essential for dealing with polynomial ideals. For instance, it determines an order in a multivariate polynomial division algorithm. Of particular interest is the {\em graded reverse lexicographic (grevlex) ordering}. 
 
\begin{definition}\label{def:grevlex}
For $\bm{\alpha}=\zag{\alpha_1,\alpha_2,\ldots,\alpha_n}$, $\bm{\beta}=\zag{\beta_1,\beta_2, \ldots, \beta_n} \in\Z_{\geq 0}^n$,  
we write
$\mbf{x}^{\bm{\alpha}}>_{grevlex}\mbf{x}^{\bm{\beta}}$ (or $\bm{\alpha} >_{grevlex} \bm{\beta}$) if $\aps{\bm{\alpha}}> \aps{\bm{\beta}}$ or $\aps{\bm{\alpha}}= \aps{\bm{\beta}}$ and the rightmost nonzero entry of $\bm{\alpha}-\bm{\beta}$ is negative.
\end{definition}
Once a monomial ordering is fixed, the meaning of leading monomial, leading term and leading coefficient of a polynomial (see Section~\ref{IntrodTheta}) is
well-defined. For more information on monomial orderings, we refer the interested reader to \cite{colios07,colios05}.

A Gr{\"o}bner basis is a particular kind of generating set of a polynomial ideal. It was first introduced in 1965 in the Phd thesis of Buchberger \cite{bu06-2}. 

\begin{definition}[Gr{\"o}bner basis]\label{groebner}
For a fixed monomial order, a basis $\mathcal{G}=\{g_1, \ldots, g_s\}$ of a polynomial ideal $J \subset \R\uglate{\mbf{x}}$ is a \textit{Gr{\"o}bner basis} (or standard basis) if for all $f \in \R\uglate{\mbf{x}}$ there exist a \textbf{unique} $r \in \R\uglate{\mbf{x}}$ and $g \in J$ such that $$f=g+r $$ and no monomial of $r$ is divisible by any of the leading monomials in $\mathcal{G}$, i.e., by any of the monomials $\LM\zag{g_1}, \LM\zag{g_2}, \ldots, \LM\zag{g_s}$.
\end{definition} 
A Gr{\"o}bner basis is not unique, but the reduced version defined next is.
\begin{definition}\label{reducedGr}
The \textit{reduced Gr{\"o}bner basis} for a polynomial ideal $J \in \R\uglate{\mbf{x}}$ is a Gr{\"o}bner basis $\mathcal{G}=\skup{g_1,g_2,\ldots,g_s}$ for $J$ such that
\begin{enumerate}
\item[1)] $\LC(g_i)=1$, for all $i \in \uglate{s} $.
\item[2)] $g_i$ does not belong to $\interval{\LT(\mathcal{G}\backslash\{g_i\})}$ for all $i \in \uglate{s}$.
\end{enumerate}
\end{definition}
In other words, a Gr{\"o}bner basis $\mathcal{G}$ is the reduced Gr{\"o}bner basis if for all $i \in \uglate{s}$ the polynomial $g_i \in \mathcal{G}$ 
is monic (i.e., $\LC(g_i)=1$) and the leading monomial $\LM(g_i)$ does not divide any monomial of $g_j$, $j\neq i$.
  
Many important properties of the ideal and the corresponding algebraic variety can be deduced via its (reduced) Gr{\"o}bner basis. 
For example, a polynomial belongs to a given ideal if and only if the unique $r$ from the Definition~\ref{groebner} equals zero.
Gr{\"o}bner bases are also one of the main computational tools in solving systems of polynomial equations \cite{colios07}.

With $\overline{f}^F$ we denote the remainder on division of $f$ by the ordered $k$-tuple $F=\zag{f_1,f_2,\ldots,f_k}$. If $F$ is a Gr{\"o}bner basis for an ideal $\interval{f_1,f_2,\ldots,f_k}$, then we can regard $F$ as a set without any particular order by Definition \ref{groebner}, or in other words, 
the result of the division algorithm does not depend on the order of the polynomials. Therefore, $\overline{f}^{\mathcal{G}}=r$ in Definition~\ref{groebner}.

The following result follows directly from Definition \ref{groebner} and the polynomial division algorithm \cite{colios07}.
\begin{corollary}
Fix a monomial ordering and let $\mathcal{G}=\{g_1,g_2,\ldots,g_s\} \subset \R\uglate{\mbf{x}}$ be a Gr{\"o}bner basis of a polynomial ideal $J$. A polynomial $f \in \R\uglate{\mbf{x}}$ is in the ideal $J$ if it can be written in the form $f=a_1 g_1+ a_2 g_2 +\ldots + a_s g_s$, where $a_i \in \R\uglate{\mbf{x}}$, for all $i \in \uglate{s}$, s.t.\ whenever $a_i g_i \neq 0$ we have $$\multideg\zag{f} \geq \multideg\zag{a_ig_i}.$$
\end{corollary}

\begin{definition}\label{def:reducezero}
Fix a monomial order and let $\mathcal{G}=\skup{g_1,g_2,\ldots,g_s} \subset \R\uglate{\mbf{x}}$. Given $f \in \R\uglate{\mbf{x}}$, we say that $f$ reduces to zero modulo $\mathcal{G}$ and write $$ f \rightarrow_{\mathcal{G}} 0$$ if it can be written in the form $f=a_1 g_1+ a_2 g_2 +\ldots + a_k g_k$ with $a_i \in \R\uglate{\mbf{x}}$ for all $i \in \uglate{k}$ s.t.\ whenever $a_i g_i \neq 0$ we have $\multideg\zag{f} \geq \multideg\zag{a_ig_i}$.
\end{definition}
Assume that $\mathcal{G}$ in the above definition is a Gr{\"o}bner basis of a given ideal $J$. Then a polynomial $f$ is in the ideal $J$ if and only if $f$ reduces to zero modulo $\mathcal{G}$. In other words, for a Gr{\"o}bner basis $\mathcal{G}$, 
$$f \rightarrow_{\mathcal{G}} 0 \hspace{0.5cm}\text{  if and only if  } \hspace{0.5cm}\overline{f}^{\mathcal{G}}=0.$$

The Gr{\"o}bner basis of a polynomial ideal always exists and can be computed in a finite number of steps via Buchberger's  algorithm \cite{bu06-2,colios07,colios05}.

Next we define the $S$-polynomial of given polynomials $f$ and $g$ which is important 
for checking whether a given basis of the ideal is a Gr{\"o}bner basis.

\begin{definition}\label{Spoly}
Let $f,g \in \R\uglate{\mbf{x}}$ be a non-zero polynomials.
\begin{enumerate}
\item If $\multideg\zag{f}=\bm{\alpha}$ and $\multideg\zag{g}=\bm{\beta}$, then let $\bm{\gamma}=\zag{\gamma_1,\gamma_2,\ldots, \gamma_n}$, where $\gamma_i=\max\skup{\alpha_i,\beta_i}$, for every $i$. We call $\mbf{x}^{\bm{\gamma}}$ the least common multiple of $\LM\zag{f}$ and $\LM\zag{g}$ written $\mbf{x}^{\bm{\gamma}}=\LCM\zag{\LM\zag{f},\LM\zag{g}}$.
\item The $S$-polynomial of $f$ and $g$ is the combination
\begin{equation*}
S\zag{f,g}=\frac{\mbf{x}^{\bm{\gamma}}}{\LT\zag{f}}f - \frac{\mbf{x}^{\bm{\gamma}}}{\LT\zag{g}}g.
\end{equation*}
\end{enumerate}
\end{definition}

The following theorem gives a criterion for checking whether a given basis of a polynomial ideal is a Gr{\"o}bner basis.

\begin{theorem}[Buchberger's criterion]\label{GrS}
A basis $\mathcal{G}=\skup{g_1,g_2,\ldots,g_s}$ for a polynomial ideal $J \subset \R\uglate{\mbf{x}}$ is a Gr{\"o}bner basis if and only if $S\zag{g_i,g_j} \rightarrow_{\mathcal{G}} 0$ for all $i \neq j$.
\end{theorem}

Computing whether $S\zag{g_i,g_j} \rightarrow_{\mathcal{G}} 0$ for all possible pairs of polynomials in the basis $\mathcal{G}$ can be a tedious task. The following proposition tells us for which pairs of polynomials this is not needed.

\begin{proposition}\label{relprime}
Given a finite set $\mathcal{G} \subset \R\uglate{\mbf{x}}$, suppose that the leading monomials of $f,g \in \mathcal{G}$ 
are relatively prime, i.e.,
\begin{equation*}
\LCM\zag{\LM\zag{f}, \LM\zag{g}}=\LM\zag{f}\LM\zag{g},
\end{equation*}
then $S\zag{f,g}\rightarrow_{\mathcal{G}} 0$.
\end{proposition}

Therefore, to prove that the set $\mathcal{G} \subset \R\uglate{\mbf{x}} $ is a Gr{\"o}bner basis, it is enough to show that $S\zag{g_i,g_j}\rightarrow_{\mathcal{G}} 0$ for those $i <j$ where $\LM\zag{g_i}$ and $\LM\zag{g_j}$ are not relatively prime. 

\section*{Acknowledgements} 
We would like to thank Bernd Sturmfels, Daniel Plaumann, and Shmuel Friedland for helpful discussions and useful inputs to this paper. We would also like to thank James Saunderson and Hamza Fawzi for the deeper insights for the matrix case scenario. We acknowledge funding by the European Research Council through the grant StG 258926 and support by the Hausdorff Research Institute for Mathematics, Bonn through the trimester program Mathematics of Signal Processing. Most of the research was done while \v{Z}.\ S.\ was a PhD student at  the University of Bonn and employed at RWTH Aachen University.

\bibliographystyle{abbrv}
\bibliography{TensorBib}

\end{document}